\documentclass[11pt,a4paper]{article}
\usepackage[margin=27mm]{geometry}
\usepackage[T1]{fontenc}
\usepackage{lmodern}
\usepackage[nopatch=footnote]{microtype}
\usepackage{amsmath,amssymb,mathtools,bm}
\usepackage{amsthm}
\usepackage{booktabs,array,tabularx,longtable}
\usepackage{enumitem}
\usepackage{xcolor}
\usepackage{xurl}
\usepackage{setspace}
\usepackage{hyperref}
\usepackage{cleveref}
\usepackage{authblk}
\usepackage{csquotes}

\definecolor{draftblue}{RGB}{30,76,130}
\hypersetup{colorlinks=true,linkcolor=draftblue,citecolor=draftblue,urlcolor=draftblue,pdftitle={A Sharpened Entropy Principle for Two-Fluid Polymer Thermodynamics},pdfauthor={Dieter Bothe}}
\newtheorem{principle}{Principle}

\newtheorem{proposition}{Proposition}
\newtheorem{theorem}{Theorem}

\newtheorem{remark}{Remark}
\newtheorem{definition}{Definition}

\newcommand{\R}{\mathbb R}
\newcommand{\I}{\mathbf I}

\newcommand{\T}{\mathbf T}
\newcommand{\D}{\mathbf D}
\newcommand{\W}{\mathbf W}
\newcommand{\Lgrad}{\mathbf L}
\newcommand{\C}{\mathbf C}
\newcommand{\taup}{\bm\tau_{p}}
\newcommand{\Jr}{\mathbf J_r}

\newcommand{\dd}{\mathrm d}
\newcommand{\Div}{\operatorname{div}}
\newcommand{\tr}{\operatorname{tr}}
\newcommand{\sym}{\operatorname{sym}}
\newcommand{\dev}{\operatorname{dev}}

\newcommand{\Db}{\frac{D}{Dt}}
\newcommand{\Zc}{\mathbf Z}
\newcommand{\CV}{\mathcal C_V}
\newcommand{\Dp}{\frac{D_p}{Dt}}

\newcommand{\eqdef}{\mathrel{:=}}

\title{\textbf{A Sharpened Entropy Principle for Two-Fluid Polymer Thermodynamics}\\[0.7em]
\large\normalfont\itshape Dedicated to the remembrance of John C. Slattery.}
\author[1]{Dieter Bothe}
\affil[1]{Department of Mathematics, Technical University of Darmstadt, Germany}
\date{}

\begin{document}
\maketitle

\begin{abstract}
A compressible, non-isothermal dilute polymer solution is formulated as a Class-II binary mixture with separate constituent mass and momentum balances and a single mixture energy and entropy balance. An additional balance-anchoring postulate restricts the constitutive factorization of selected dissipative mechanisms in a fixed balance representation.

A connector population balance fixes polymer convection and, by second moments, yields a two-velocity upper-convected conformation rate. For a density-linear configurational free energy, conformation transport cancels its chemical-potential counterpart, while deformation power cancels the elastic partial-stress power when their weights match. Within this mechanism-wise constitutive framework, a Gordon--Schowalter modification is compatible with the unchanged Kramers stress only for the affine choice, unless another reversible channel is supplied. For Hookean springs, the moment equation closes without Gaussianity; explicit nonnegative remainder identities quantify the free energy and dissipation not resolved by the conformation tensor.

An entropy-invariant Class-II-to-Class-I constitutive identification yields thermo-chemical, configurational-stress and partial-viscous-stress diffusion forces. Eliminating relative acceleration is distinguished from omitting the quadratic relative-inertia momentum flux and its kinetic-energy storage. Coordinated stress--interaction changes preserve the momentum and energy balances but shift the local entropy pair by a divergence. The temperature equations distinguish relaxation heating from entropy production and retain the chain-rule terms generated by temperature-dependent conformation normalization. Classical Oldroyd-B/UCM rheology is recovered only after further one-velocity, incompressible and isothermal reductions.
\end{abstract}


\section{Introduction}

On 24 February 2007, I wrote to John C. Slattery seeking advice on surfactant transport, surface viscosity and the differential geometry of interfacial balances. In his reply, he directed me to the second edition of \emph{Interfacial Transport Phenomena}, coauthored with Leonard Sagis and Eun-Suok Oh~\cite{SlatterySagisOh2007}, and warmly invited further questions. He emphasized that neglecting surfactant mass transfer to the adjacent phases is an approximation, even when it is appropriate for the problem at hand.\footnote{Personal correspondence with John C. Slattery, in reply to the author's inquiry of 24 February 2007.} His response combined scientific clarity with generosity toward a colleague working between mathematics and chemical engineering. I dedicate this contribution to his memory with gratitude.

The systematic distinction between balance laws and constitutive assumptions is central to the continuum treatment of momentum, energy and mass transport exemplified by Slattery's \emph{Advanced Transport Phenomena}~\cite{Slattery1999}. Although the present study concerns a bulk solvent--polymer mixture rather than an interface, the distinction between a physical mechanism and its neglect in a reduced model is equally pertinent. The thermodynamic two-fluid formulation developed below therefore keeps balance identities, constitutive choices and model reductions explicit.

Class-II mixture theory retains a separate momentum balance for each constituent, whereas Class-I theory replaces the partial momenta by a barycentric momentum balance while retaining partial mass balances. In a binary polymer solution this distinction exposes the relative velocity, interconstituent momentum production and separate constituent stress powers before reduction. These quantities participate in the mechanical exchange between solvent, polymer and polymer configuration and are therefore relevant to the local entropy exploitation~\cite{BotheDreyer2015}.

The entropy principle used below is the total entropy balance of the complete modeled system. Its exploitation is supplemented by Principle~V, introduced by the author in the present work, which requires a constitutive cofactor in each selected binary dissipative product to be anchored in the unclosed balances. This is an additional constitutive postulate, not a consequence of the scalar second-law inequality alone. This requirement is imposed after sign-indefinite reversible powers have been identified and removed from the entropy production.

Polymer configuration is represented first by a connector distribution and subsequently by a conformation tensor. The transport kinematics are not selected by objectivity alone: polymer centers are carried with $\bm v_p$ in the connector population balance, while connector deformation is generated by $\Lgrad_d$. Taking the second moment therefore fixes the polymer material derivative and yields the two-velocity upper-convected conformation operator; observer covariance is then verified as a consequence. A reduced configurational relative-free-energy balance contains a signed mechanical exchange term. In the Class-II formulation this term is paired with the complementary constituent stress power and cancels before the total entropy production is decomposed. At conformation level that cancellation constrains deformation velocity and elastic-stress decomposition, and the later Gordon--Schowalter test shows that, for the stated dumbbell free energy and Kramers stress with no additional reversible channel, the affine upper-convected rate is the compatible choice.

The Class-II model is subsequently reduced algebraically to a Class-I polymer mixture. Partial mass and internal-energy balances are identified, the Class-I entropy flux is selected so that $\zeta^{\rm I}\equiv\zeta^{\rm II}$, and the Class-I diffusion force inherits thermo-chemical, configurational-stress and partial-viscous-stress contributions. The reduction distinguishes three operations: a change of variables, a constitutive identification eliminating relative acceleration, and omission of quadratic relative inertia from the reduced momentum and energy architecture. On parent processes satisfying the identification, the quadratic momentum flux is paired with a reversible relative kinetic-energy storage and transport term. This conditional identity does not make the full relative dynamics reversible or establish dynamical equivalence of the two model classes. A zero-diffusion limit is taken only afterward.

The non-isothermal reductions are compared with established viscoelastic energy equations in \Cref{sec:temperature-equation}; the broader two-fluid polymer literature is discussed in \Cref{sec:literature-comparison}.

Several ingredients used below have established precedents: Kramers stress~\cite{Bird1987,Ottinger1996}, two-fluid stress--composition coupling~\cite{DoiOnuki1992,Apostolakis2002,ZhouDoi2022}, non-isothermal conformation thermodynamics~\cite{WapperomHulsen1998,Hron2017,Dostalik2020}, and entropy-invariant Class-II to Class-I reduction for molecular mixtures~\cite{BotheDreyer2015}. The contribution of the present work is their coupled Class-II polymer synthesis and the structural results obtained from that synthesis: explicit cancellation of configurational and elastic partial-stress powers within the total entropy identity in a fixed balance representative; a microstructural derivation and thermodynamic compatibility test for the two-velocity upper-convected conformation rate; an explicit characterization of the representation dependence of local mechanism-wise production under coordinated stress--interaction changes; and an entropy-invariant Class-II-to-Class-I reduction that selects a descendant entropy flux preserving the parent production while exposing the inherited diffusion-force structure and distinguishing relative-acceleration elimination from quadratic-inertia truncation.

Throughout, the qualifier ``exact'' is used relative to the explicitly stated balance architecture, constitutive choices and state variables. Thus the entropy and representation identities are algebraic identities of the Class-II model. At the kinetic distribution level, the population and free-energy identities are exact under their stated regularity and boundary assumptions, while nonnegative connector production follows after the connector-flux law has been selected. For Hookean springs the second-moment evolution closes exactly at the level of $\C$ without a Gaussian assumption; Gaussianity is needed for equality of the kinetic and finite-dimensional free energies and dissipations; below, explicit nonnegative remainder identities quantify their difference for non-Gaussian distributions. The Class-I constitutive map additionally eliminates relative acceleration, and its standard momentum/energy architecture omits quadratic relative inertia. This terminology is used to keep balance identities, constitutive assumptions, closure assumptions and algebraic reductions logically distinct.

\section{The sharpened entropy principle: total-system entropy balance and balance anchoring}
\label{sec:entropy-principle}

The second-law statement used throughout is the entropy balance of the complete solvent--polymer system. No separate nonnegative entropy production is postulated for the constituents. Reduced subsystem identities, including relative-free-energy balances, may contain signed exchange terms and are not assigned independent entropy inequalities. Mechanism-wise nonnegativity is imposed only on the total entropy production after reversible internal exchanges have been combined.

The first four principles below follow the axiomatic formulation of Bothe--Dreyer~\cite{BotheDreyer2015}, adapted to the present notation, a volumetric entropy density and the total solvent--polymer entropy balance. Principle~V is the additional balance-anchoring axiom introduced here. In the standard bulk balances used below, constitutive quantities occur as additive parts of unclosed fluxes and sources; Principle~V identifies one such constitutive quantity in every selected binary entropy-production product.

\begin{definition}[Admissible process and local thermodynamic test]
An \emph{admissible process} is a sufficiently smooth solution of the balances, kinematic constraints and constitutive restrictions imposed at the stage under consideration. A \emph{local thermodynamic test} is a variation of the state and the derivatives entering those relations at a point, subject to the relations already imposed. Independence of particular derivatives is an explicit richness hypothesis of the constitutive exploitation, not an existence theorem for the field equations. In particular, the sign-variation tests below assume that the stated deformation components can be varied independently of the fixed state. No assertion is made that every formally compatible local test extends to a solution of an initial-boundary-value problem.
\end{definition}

\begin{principle}[Entropy/entropy-flux pair]
There is an entropy/entropy-flux pair $(s,\bm\Phi)$ as material-dependent quantities, where $s$ is an objective scalar and $\bm\Phi$ is an objective vector.
\end{principle}

\begin{principle}[Entropy balance]
The pair $(s,\bm\Phi)$ satisfies
\begin{equation}
 \partial_t s+\Div(s\bm v+\bm\Phi)=\zeta,
 \label{eq:total-entropy-balance}
\end{equation}
where $\zeta$ is the total entropy production.
\end{principle}

\begin{principle}[Parity]
The parity of a physical quantity is determined from its physical dimension in SI base units: it is negative if the unit of time, second, occurs with an odd power and positive if it occurs with an even power. This is the dimensional parity convention for the present non-electromagnetic balance theory; it is not a statement that dissipative solutions are invariant under time reversal. Hence the volumetric entropy density $s$ has positive parity, whereas the entropy flux $\bm\Phi$ and the entropy production $\zeta$ have negative parity.
\end{principle}

\begin{principle}[Binary dissipative mechanisms]
Any admissible entropy flux is such that
\begin{enumerate}[label=(\roman*)]
\item the entropy production consists of a sum of binary products,
\begin{equation}
 \zeta=\sum_m N_mP_m,
 \label{eq:binary-entropy}
\end{equation}
where the $N_m$ denote quantities of negative parity and the $P_m$ quantities of positive parity. Products denote the appropriate scalar, vector or tensor contractions, or duality pairings of specified coupled blocks;
\item each binary product describes a dissipative mechanism, or a coupled block of mechanisms, introduced in advance, and
\begin{equation}
 N_mP_m\ge0
 \label{eq:mechanism-nonnegative}
\end{equation}
for all $m$ and every admissible local thermodynamic process.
\end{enumerate}
\end{principle}

\begin{principle}[Balance anchoring]
\label{pr:V}
Fix a balance representation and its primitive state variables. For a bulk balance
\[
 \partial_t a+\Div\bm J=r,
\]
a \emph{balance anchor} is an additive constitutive part of $\bm J$ or $r$ that is not fixed kinematically or thermodynamically before constitutive closure. Every selected binary dissipative product $N_mP_m$ contains such an anchor; for a coupled block, the selected constitutive factors are collected into an anchored vector.
\end{principle}

Principle~V restricts refactorization after the balance representative and the principal mechanisms have been fixed. It does not uniquely select those mechanisms, their grouping, a stress representative or an entropy flux. Parity is attached to the selected anchored factors, rather than assigned anew after an arbitrary factorization of the same scalar product. Rescaling a selected factor by a nonvanishing even state scalar, with inverse rescaling of its partner, leaves its parity and production unchanged. More general entropy-invariant mixing must be specified explicitly. Positivity alone constrains the symmetric part of a linear constitutive block; no independent derivation of microscopic Onsager--Casimir reciprocity from Principle~V alone is claimed.

The mechanism-wise condition above does not exclude cross-effects. Principal mechanisms are first selected with their anchored factors; entropy-invariant mixing may subsequently combine several such mechanisms into an Onsager block. Positivity is then imposed on the resulting block quadratic form, not on each individual cross-term after expansion. This is the sense in which mechanism-wise nonnegativity and entropy-invariant cross-coupling are used below. For further details on this axiomatic form of the entropy principle and its extension to sharp-interface multicomponent systems, see Bothe--Dreyer~\cite{BotheDreyer2015} and Bothe~\cite{BotheSharpInterface2022,BotheMultiVelocity2025}.

\section{Class-II parent model: a compressible non-isothermal binary solution}
\label{sec:mixture-architecture}

At Class-II level, each constituent has a mass and momentum balance, while energy and entropy are balanced for the mixture. The separate momentum balances retain the interconstituent momentum production and constituent stress powers needed below to identify reversible configurational exchange before reduction to a barycentric momentum equation.

Consider two molecular constituents, solvent $s$ and polymer $p$. The polymer species is a prototype representing one population of long chains. No incompressibility or isothermal constraint is imposed at the constitutive starting point. Spatial dimension is denoted by $d$ (physically $d=3$). We use $(\nabla\bm v)_{ij}=\partial_jv_i$, $(\Div\T)_i=\partial_jT_{ij}$ and $\mathbf A:\mathbf B=\tr(\mathbf A^{\sf T}\mathbf B)$, and set $\dev\mathbf A=\mathbf A-(\tr\mathbf A)\I/d$. A constituent subscript $\alpha\in\{s,p\}$ is distinct from the unindexed scalar deformation weight $\alpha$ introduced below.

Let
\begin{equation}
 \rho_s>0,\quad \rho_p>0,\qquad \bm v_s,\quad \bm v_p,\qquad T>0
 \label{eq:basic-fields}
\end{equation}
be the partial mass densities, constituent velocities and common local temperature. The local entropy exploitation is restricted to the interior of the state space, where both constituent densities are positive; limiting states with a vanishing polymer or solvent density require a separate degenerate-limit analysis. The common-temperature assumption expresses fast thermal equilibration between the molecular species but does \emph{not} impose $T=\mathrm{const}$. Define
\begin{equation}
 \rho=\rho_s+\rho_p,
 \qquad
 \bm v=\frac{\rho_s\bm v_s+\rho_p\bm v_p}{\rho},
 \qquad
 \bm w=\bm v_p-\bm v_s.
 \label{eq:barycentric}
\end{equation}
The partial mass balances are
\begin{equation}
 \partial_t\rho_\alpha+\Div(\rho_\alpha\bm v_\alpha)=0,
 \qquad \alpha\in\{s,p\}.
 \label{eq:partial-mass}
\end{equation}
No chemical conversion between solvent and polymer is considered.

The partial momentum balances are written as
\begin{equation}
 \partial_t(\rho_\alpha\bm v_\alpha)
 +\Div(\rho_\alpha\bm v_\alpha\otimes\bm v_\alpha)
 =\Div\T_\alpha+\bm m_\alpha,
 \label{eq:partial-momentum}
\end{equation}
with
\begin{equation}
 \bm m_s+\bm m_p=\bm0.
 \label{eq:momentum-exchange}
\end{equation}
Here $\bm m_\alpha$ is the interconstituent momentum production. External body forces are omitted for clarity; they may be included provided their mechanical power is incorporated consistently in the total-energy balance, in which case the internal-energy identities below are unchanged after that external power is accounted for. Couple stresses and body couples are excluded, and as an additional constitutive restriction the constituent Cauchy stresses are taken symmetric, $\T_\alpha=\T_\alpha^{\sf T}$. This symmetry assumption is what permits the mechanical power to be written as $\T_\alpha:\D_\alpha$ rather than $\T_\alpha:\Lgrad_\alpha$; the formulation does not rely on deriving symmetry of each partial stress from the total angular-momentum balance alone.

Use
\begin{equation}
 \Lgrad_\alpha=\nabla\bm v_\alpha,
 \qquad
 \D_\alpha=\sym\Lgrad_\alpha,
 \qquad \alpha\in\{s,p\}.
 \label{eq:Dalpha}
\end{equation}

The model contains one total energy balance. With
\begin{equation}
 K=\frac12\rho_s|\bm v_s|^2+\frac12\rho_p|\bm v_p|^2,
 \qquad E=u+K,
 \label{eq:energy-split}
\end{equation}
write
\begin{equation}
 \partial_t E+\Div\bm J_E=0.
 \label{eq:total-energy}
\end{equation}
Here $u$ is the total internal-energy density of the solvent--polymer mixture, including configurational energy. Multiplication of the partial momentum balances by the constituent velocities and subtraction from \eqref{eq:total-energy} gives the exact internal-energy equation
\begin{equation}
 \partial_t u+\Div\bm J_u
 =\sum_{\alpha=s,p}\T_\alpha:\D_\alpha-\bm m_p\cdot\bm w,
 \label{eq:internal-energy-exact}
\end{equation}
where
\begin{equation}
 \bm J_u
 =\bm J_E-\sum_{\alpha=s,p}
 \left(\frac12\rho_\alpha|\bm v_\alpha|^2\bm v_\alpha-\T_\alpha\bm v_\alpha\right).
 \label{eq:internal-energy-flux}
\end{equation}
Thus the relative-motion interaction power is part of the \emph{single total} energy balance; it is not an external heat source.

Define the barycentric diffusion fluxes
\begin{equation}
 \bm J_\alpha=\rho_\alpha(\bm v_\alpha-\bm v),
 \qquad \bm J_s+\bm J_p=\bm0,
 \label{eq:mass-diffusion-flux}
\end{equation}
and decompose
\begin{equation}
 \bm J_u=u\bm v+\bm j_u.
 \label{eq:ju-bary}
\end{equation}

\subsection{Thermodynamic state and reduced entropy production rate}
\label{sec:exact-entropy}

At the conformation-tensor level the local Helmholtz free-energy density is taken as
\begin{equation}
 f=f(\rho_s,\rho_p,T,\C),
 \qquad \C=\C^{\sf T}>0\quad\text{(symmetric positive definite)}.
 \label{eq:full-free-energy-state}
\end{equation}
On the interior state domain used below $f$ is assumed to be at least $C^2$ in its arguments, with the additional smoothness required when higher temperature derivatives are invoked in the specialized thermal reductions.
Introduce
\begin{align}
 s&=-\frac{\partial f}{\partial T},
 &\mu_\alpha&=\frac{\partial f}{\partial\rho_\alpha},
 &\Zc&=\frac{\partial f}{\partial\C},
 \label{eq:thermo-conjugates}\\
 u&=f+Ts,
 &p&=\rho_s\mu_s+\rho_p\mu_p-f.
 \label{eq:u-p-def}
\end{align}
Here $s$ is the entropy per physical volume, $\mu_\alpha$ are mass-specific chemical potentials, and $\Zc$ is the tensor thermodynamic force conjugate to conformation. The Gibbs relation is
\begin{equation}
 \dd u=T\,\dd s+\mu_s\,\dd\rho_s+\mu_p\,\dd\rho_p+\Zc:\dd\C.
 \label{eq:gibbs-u}
\end{equation}
In the chosen density variables define mass-specific entropy and internal-energy derivative coefficients by
\begin{equation}
 \bar s_\alpha=-\frac{\partial\mu_\alpha}{\partial T},
 \qquad
 h_\alpha=\mu_\alpha+T\bar s_\alpha,
 \label{eq:partial-entropy-enthalpy}
\end{equation}
All derivatives in \eqref{eq:partial-entropy-enthalpy} hold the other partial density and $\C$ fixed. Thus $h_\alpha=(\partial u/\partial\rho_\alpha)_{T,\rho_{\beta\ne\alpha},\C}$ is an internal-energy derivative coefficient in these coordinates; it is not asserted to be the conventional partial specific enthalpy defined at fixed pressure. The overbar distinguishes the mass-specific coefficient $\bar s_\alpha$ from the volumetric entropy density $s$. These coefficients specify a heat/relative-energy transport convention. Define the corresponding reduced heat flux $\bm q$ through
\begin{equation}
 \bm q=\bm j_u-\sum_{\alpha=s,p}h_\alpha\bm J_\alpha.
 \label{eq:reduced-heat-flux}
\end{equation}
The corresponding nonconvective entropy flux is
\begin{equation}
 \bm\Phi=\sum_{\alpha=s,p}\bar s_\alpha\bm J_\alpha+\frac{\bm q}{T}.
 \label{eq:entropy-flux-exact}
\end{equation}
Thus the total entropy flux appearing under the divergence in Principle~II is $s\bm v+\bm\Phi$. The nonconvective internal-energy flux is partitioned into the selected reduced heat flux and relative constituent transport with the coefficients $h_\alpha$. This convention must be preserved when comparing heat-flux closures. Once the entropy flux \eqref{eq:entropy-flux-exact} has been selected, the entropy balance determines the \emph{reduced entropy production rate}; ``reduced'' refers here to the prior selection of the entropy flux, not to the later Class-II-to-Class-I reduction.

\begin{proposition}[Reduced entropy production rate]
\label{prop:exact-entropy}
Assume \eqref{eq:partial-mass}, \eqref{eq:partial-momentum}, \eqref{eq:internal-energy-exact} and the Gibbs structure \eqref{eq:thermo-conjugates}--\eqref{eq:gibbs-u}. Then, for the selected entropy flux \eqref{eq:entropy-flux-exact}, the reduced entropy production rate is
\begin{align}
 T\zeta={}&
 \sum_{\alpha=s,p}\T_\alpha:\D_\alpha
 -\bm m_p\cdot\bm w
 +p\Div\bm v
 -\Zc:\Db\C \notag\\
 &-\sum_{\alpha=s,p}\bm J_\alpha\cdot
 \left(\nabla\mu_\alpha+\bar s_\alpha\nabla T\right)
 -\frac{1}{T}\bm q\cdot\nabla T,
 \label{eq:exact-entropy-identity}
\end{align}
where $\Db=\partial_t+\bm v\cdot\nabla$ is the barycentric material derivative.
\end{proposition}

\begin{proof}
Write \eqref{eq:internal-energy-exact} as
\[
 \Db u+u\Div\bm v+\Div\bm j_u
 =\sum_\alpha\T_\alpha:\D_\alpha-\bm m_p\cdot\bm w.
\]
The Gibbs relation gives
\[
 T\Db s=\Db u-\sum_\alpha\mu_\alpha\Db\rho_\alpha-\Zc:\Db\C.
\]
Using
\[
 \Db\rho_\alpha=-\rho_\alpha\Div\bm v-\Div\bm J_\alpha
\]
and $\partial_t s+\Div(s\bm v)=\Db s+s\Div\bm v$, one obtains
\begin{align*}
 T\big[\partial_t s+\Div(s\bm v)\big]
 ={}&\sum_\alpha\T_\alpha:\D_\alpha-\bm m_p\cdot\bm w
 -\Div\bm j_u-\Zc:\Db\C\\
 &+\big(-u+Ts+\sum_\alpha\rho_\alpha\mu_\alpha\big)\Div\bm v
 +\sum_\alpha\mu_\alpha\Div\bm J_\alpha.
\end{align*}
Because $u=f+Ts$, the coefficient of $\Div\bm v$ is exactly
$\sum_\alpha\rho_\alpha\mu_\alpha-f=p$. Moreover,
\[
 \sum_\alpha\mu_\alpha\Div\bm J_\alpha
 =\Div\!\left(\sum_\alpha\mu_\alpha\bm J_\alpha\right)
 -\sum_\alpha\bm J_\alpha\cdot\nabla\mu_\alpha.
\]
By \eqref{eq:reduced-heat-flux},
\[
 \bm j_u-\sum_\alpha\mu_\alpha\bm J_\alpha
 =\bm q+T\sum_\alpha\bar s_\alpha\bm J_\alpha.
\]
Finally use
\begin{align*}
 T\Div\!\left(\sum_\alpha\bar s_\alpha\bm J_\alpha\right)
 &=\Div\!\left(T\sum_\alpha\bar s_\alpha\bm J_\alpha\right)
 -\sum_\alpha\bar s_\alpha\bm J_\alpha\cdot\nabla T,\\
 T\Div\!\left(\frac{\bm q}{T}\right)
 &=\Div\bm q-\frac{1}{T}\bm q\cdot\nabla T.
\end{align*}
Collecting these two divergence contributions according to the already selected entropy flux \eqref{eq:entropy-flux-exact} gives precisely \eqref{eq:exact-entropy-identity}. Thus \eqref{eq:exact-entropy-identity} is the reduced entropy-production identity at this stage; constitutive sign restrictions are imposed only in the subsequent exploitation.
\end{proof}

For the dilute conformation model, the configurational free energy is homogeneous of degree one in polymer density,
\begin{equation}
 f(\rho_s,\rho_p,T,\C)
 =f_0(\rho_s,\rho_p,T)
 +\frac{\rho_p}{m_p^{\mathrm{mol}}}\,a_C(T,\C).
 \label{eq:dilute-f-split}
\end{equation}
Consequently the configurational part contributes to $\mu_p$, entropy and heat capacity, but cancels from the thermodynamic pressure $p$; this cancellation is used in the later pressure--temperature reduction.

Split the polymer chemical-potential gradient into its ordinary thermo-compositional part and its explicit conformation part:
\begin{equation}
 \nabla\mu_p
 =\nabla^{\circ}\mu_p
 +\frac{1}{\rho_p}\,\Zc:\nabla\C,
 \qquad
 \nabla\mu_s=\nabla^{\circ}\mu_s,
 \label{eq:mu-circle}
\end{equation}
where $\nabla^{\circ}$ denotes a spatial derivative at fixed $\C$ and the vector $(\Zc:\nabla\C)_i=\Zc:\partial_i\C$.

\begin{proposition}[Exact cancellation of conformation transport in the entropy identity]
\label{prop:convective-cancellation}
For the dilute free energy \eqref{eq:dilute-f-split}, the conformation-dependent part of polymer chemical-potential transport cancels exactly the difference between barycentric and polymer convection of $\C$:
\begin{equation}
 -\Zc:\Db\C
 -\bm J_p\cdot\frac{1}{\rho_p}\Zc:\nabla\C
 =-\Zc:\Dp\C.
 \label{eq:convective-cancellation}
\end{equation}
\end{proposition}

\begin{proof}
Because $\bm v_p-\bm v=\bm J_p/\rho_p$,
\[
 \Db\C=\Dp\C-\frac{\bm J_p}{\rho_p}\cdot\nabla\C.
\]
Substitution gives two equal and opposite contractions with $\bm J_p\cdot(\Zc:\nabla\C)/\rho_p$. No constitutive closure is involved.
\end{proof}

Relative transport of configurational free energy is already contained in the polymer chemical-potential/relative-energy transport and is therefore not an additional dissipative mechanism.

\subsection{Explicit reversible extraction of the Class-II mixture part}
\label{sec:explicit-mixture-gauge}

After Proposition~\ref{prop:convective-cancellation}, the remaining chemical-potential gradients in \eqref{eq:exact-entropy-identity} are the ordinary derivatives $\nabla^{\circ}\mu_\alpha$. An explicit reversible Class-II representation is fixed next and its cancellation is verified within the same total entropy identity.

Set
\begin{equation}
 \omega_\alpha=\frac{\rho_\alpha}{\rho},
 \qquad
 M_{sp}=\frac{\rho_s\rho_p}{\rho},
 \qquad
 \bm X_\alpha^{\circ}=\nabla^{\circ}\mu_\alpha+\bar s_\alpha\nabla T.
 \label{eq:ordinary-forces}
\end{equation}
Since $\bm J_p=M_{sp}\bm w$ and $\bm J_s=-\bm J_p$, choose the ordinary reversible partial stresses and interaction force as
\begin{align}
 \T_{s,0}&=-\omega_s p\I,
 &\T_{p,0}&=-\omega_p p\I,
 \label{eq:explicit-rev-stress}\\
 \bm m_{p,0}
 &=-M_{sp}(\bm X_p^{\circ}-\bm X_s^{\circ})
   -p\nabla\omega_s,
 &\bm m_{s,0}&=-\bm m_{p,0}.
 \label{eq:explicit-rev-interaction}
\end{align}
This specifies one stress/interaction representative; the decomposition into partial stresses and interaction force is not unique. Indeed,
\begin{equation}
 \Div\bm v
 =\omega_s\Div\bm v_s+\omega_p\Div\bm v_p
 -\bm w\cdot\nabla\omega_s,
 \label{eq:div-v-mixture}
\end{equation}
and therefore
\begin{align}
 &\sum_{\alpha=s,p}\T_{\alpha,0}:\D_\alpha
 +p\Div\bm v
 -\bm m_{p,0}\cdot\bm w
 -\sum_{\alpha=s,p}\bm J_\alpha\cdot\bm X_\alpha^{\circ}=0.
 \label{eq:ordinary-reversible-cancellation}
\end{align}
Thus the ordinary pressure, diffusion and reversible interaction powers cancel identically in the total entropy identity. Explicitly, \eqref{eq:div-v-mixture} gives
\[
 \sum_\alpha\T_{\alpha,0}:\D_\alpha+p\Div\bm v
 =-p\,\bm w\cdot\nabla\omega_s,
\]
whereas \eqref{eq:explicit-rev-interaction} yields
\[
 -\bm m_{p,0}\cdot\bm w
 =M_{sp}(\bm X_p^{\circ}-\bm X_s^{\circ})\cdot\bm w
 +p\,\bm w\cdot\nabla\omega_s.
\]
Since $\bm J_p=M_{sp}\bm w=-\bm J_s$, the remaining thermo-chemical term is exactly the negative of the first term on the right. This verifies \eqref{eq:ordinary-reversible-cancellation} without invoking a constitutive closure.

Now decompose
\begin{equation}
 \T_\alpha=\T_{\alpha,0}+\T_\alpha^{\rm el}+\T_\alpha^d,
 \qquad
 \bm m_p=\bm m_{p,0}+\bm m_p^{\rm el}+\bm m_p^d,
 \label{eq:full-reversible-split}
\end{equation}
where the superscript ``el'' denotes the reversible configurational contribution treated in \Cref{sec:cancellation}. In the explicit gauge used throughout the cancellation theorem we set
\begin{equation}
 \bm m_p^{\rm el}=\bm0.
 \label{eq:elastic-interaction-gauge}
\end{equation}
The gauge $\bm m_p^{\rm el}=\bm0$ assigns the reversible elastic power to the partial elastic stresses; it is a representation choice rather than a statement of physical decoupling. Proposition~\ref{prop:gauge-transform} gives the corresponding transformation to other balance-equivalent stress/interaction representatives. The gauge is fixed in the subsequent entropy exploitation. After the ordinary cancellation \eqref{eq:ordinary-reversible-cancellation}, define the residual parent-mixture contribution
\begin{equation}
 \zeta_{\mathrm{mix}}
 =-\frac{\bm q\cdot\nabla T}{T^2}
 +\frac{1}{T}\sum_{\alpha=s,p}\T^{d}_\alpha:\D_\alpha
 -\frac{1}{T}\bm m^{d}_p\cdot(\bm v_p-\bm v_s).
 \label{eq:classII-entropy}
\end{equation}
\begin{proposition}[Coordinated stress--interaction representation change]
\label{prop:gauge-transform}
Let $\mathbf S=\mathbf S^{\sf T}$ be a sufficiently smooth tensor field and define
\begin{align}
 \T_p'&=\T_p+\mathbf S,
 &\T_s'&=\T_s-\mathbf S,\notag\\
 \bm m_p'&=\bm m_p-\Div\mathbf S,
 &\bm m_s'&=-\bm m_p',
 \label{eq:gauge-transform-mech}
\end{align}
together with
\begin{equation}
 \bm J_u'=\bm J_u+\mathbf S\bm w,
 \qquad
 \bm q'=\bm q+\mathbf S\bm w.
 \label{eq:gauge-transform-flux}
\end{equation}
Then both partial momentum balances and the internal-energy balance are unchanged as balance equations. The local mechanical-power density transforms according to
\begin{equation}
 \sum_\alpha\T_\alpha':\D_\alpha-\bm m_p'\cdot\bm w
 =\sum_\alpha\T_\alpha:\D_\alpha-\bm m_p\cdot\bm w
 +\Div(\mathbf S\bm w).
 \label{eq:gauge-power-divergence}
\end{equation}
Consequently the entropy-flux/production pair changes by
\begin{equation}
 \bm\Phi'=\bm\Phi+\frac{\mathbf S\bm w}{T},
 \qquad
 \zeta'=\zeta+\Div\!\left(\frac{\mathbf S\bm w}{T}\right).
 \label{eq:gauge-entropy-pair}
\end{equation}
Thus momentum--energy balance equivalence alone does not make the \emph{pointwise} entropy production or its mechanism-wise decomposition gauge invariant.
\end{proposition}

\begin{proof}
The two momentum right-hand sides are unchanged by direct substitution. Since $\mathbf S$ is symmetric,
\[
 \mathbf S:(\D_p-\D_s)+(\Div\mathbf S)\cdot\bm w
 =\mathbf S:\nabla\bm w+(\Div\mathbf S)\cdot\bm w
 =\Div(\mathbf S\bm w),
\]
which proves \eqref{eq:gauge-power-divergence}; \eqref{eq:gauge-transform-flux} then gives the same internal-energy balance. Because the relative-energy flux is unchanged, the reduced heat flux changes by the same $\mathbf S\bm w$. Substitution in the entropy flux and exact entropy identity gives \eqref{eq:gauge-entropy-pair}.
\end{proof}

For periodic boundaries, or when $(\mathbf S\bm w)\cdot\bm n=0$, the integrated production is unchanged by \eqref{eq:gauge-entropy-pair}, whereas the pointwise entropy production and its mechanism decomposition change. Local mechanism-wise nonnegativity is therefore imposed only after a representative and its associated entropy flux have been fixed. In this sense, ``gauge'' below means balance-representation freedom: a transformed representative is not automatically a second-law-admissible realization of the same constitutive closure unless its transformed entropy pair satisfies the required local inequality.

One diagonal closure is
\begin{align}
 \bm q&=-\mathbf K_T\nabla T,
 &&\mathbf K_T=\mathbf K_T^{\sf T}\ge0,
 \label{eq:fourier}\\
 \T^d_\alpha&=2\eta_\alpha\dev\D_\alpha
 +\kappa_\alpha^{b}(\Div\bm v_\alpha)\I,
 &&\eta_\alpha,\kappa_\alpha^{b}\ge0,
 \label{eq:partial-visc}\\
 \bm m^d_p&=-\gamma\bm w,
 &&\gamma\ge0.
 \label{eq:drag-closure}
\end{align}
The conductivity tensor $\mathbf K_T$ may depend on conformation through an objective tensor function, allowing deformation-induced anisotropic heat conduction. More general heat/diffusion and momentum cross-couplings can be represented by entropy-invariant Onsager blocks but are not included here.

\section{Configurational population balance}
\label{sec:config-balance}

This section temporarily resolves polymer configuration by the connector distribution $\varphi$ rather than by the finite-dimensional state $\C$. It supplies a microstructural population balance and a constitutive connector-space flux law from which the conformation model is obtained by moments; it is not an additional configurational state superposed on $\C$ in the same free energy. The two levels are compared explicitly in \Cref{prop:gaussian-dissipation-match}.

Let $\bm r\in\R^d$ denote the dumbbell connector vector. Let $n(t,\bm x)$ be the number density of polymer molecules and let $\varphi(t,\bm x,\bm r)$ be the normalized local connector distribution,
\begin{equation}
 \varphi\ge0,
 \qquad
 \int_{\R^d}\varphi\,\dd\bm r=1.
 \label{eq:phi-normalized}
\end{equation}
For a prototype polymer molecule of mass $m_p^{\mathrm{mol}}$,
\begin{equation}
 \rho_p=m_p^{\mathrm{mol}}n.
 \label{eq:rhop-n}
\end{equation}
Hence \eqref{eq:partial-mass} implies
\begin{equation}
 \Dp n+n\Div\bm v_p=0,
 \qquad
 \Dp=\partial_t+\bm v_p\cdot\nabla.
 \label{eq:n-balance}
\end{equation}

Polymer centers are transported with $\bm v_p$, whereas the connector may be deformed by a different local velocity gradient $\Lgrad_d$. The connector-space balance is
\begin{equation}
 \Dp\varphi
 +\nabla_{\bm r}\cdot\left(\varphi\Lgrad_d\bm r+\Jr\right)=0.
 \label{eq:phi-balance}
\end{equation}
The connector flux $\Jr$ is nonconvective in configuration space and is a constitutive quantity explicitly anchored in \eqref{eq:phi-balance}; it will be one of the Principle-V entropy factors. Equivalently, for the chain population density $\Psi=n\varphi$ one has the conservative balance in the extended $(\bm x,\bm r)$ space
\begin{equation}
 \partial_t\Psi+\Div_{\bm x}(\Psi\bm v_p)
 +\nabla_{\bm r}\cdot\left(\Psi\Lgrad_d\bm r+n\Jr\right)=0.
 \label{eq:Psi-balance}
\end{equation}
The configurational evolution is therefore a population balance rather than an additional entropy law. It is a local affine-dumbbell modeling assumption: every connector subpopulation has the same center velocity $\bm v_p$, and no independent conformation-dependent spatial flux is retained. Macroscopic relative polymer transport remains present through $\bm v_p-\bm v_s$. Finite-size nonlocal bead sampling and center--connector correlations are outside this local ansatz; the more detailed kinetic construction in~\cite{Degond2010} need not reduce to this form without additional approximations.

Consider the two-fluid family
\begin{equation}
 \Lgrad_d=\alpha\Lgrad_p+(1-\alpha)\Lgrad_s,
 \qquad
 \Lgrad_\alpha=\nabla\bm v_\alpha,
 \qquad 0\le\alpha\le1,
 \label{eq:Ld-alpha}
\end{equation}
with constant $\alpha$. The cases $\alpha=1$ and $\alpha=0$ mean that conformation is deformed by the polymer and solvent gradients, respectively. Degond--Lozinski--Owens motivate the latter choice for a dilute solution by identifying the fluid surrounding the beads with the solvent~\cite{Degond2010}. Zhou--Doi show explicitly that polymer, solvent and weighted gradients lead to different two-fluid dynamics in inhomogeneous solutions~\cite{ZhouDoi2022}.

Equation \eqref{eq:phi-balance} is specified before any tensor objective rate is selected. More specifically, the material transport is already fixed by the population kinematics: polymer centers are carried by $\bm v_p$, so moments of the balance inherit $\Dp$. The deformation part is inherited from $\Lgrad_d$. The conformation rate is therefore derived from the microstructural transport law; frame indifference is checked afterward as a covariance requirement.

\subsection{Observer covariance of the two-velocity transport}
\label{sec:objectivity-twofluid}
The use of two velocities does not compromise objectivity. Consider a superposed rigid observer change
\begin{equation}
 \bm x^*=\mathbf Q(t)\bm x+\bm c(t),
 \qquad
 \mathbf Q\mathbf Q^{\sf T}=\I,
 \label{eq:observer-change}
\end{equation}
with spin $\mathbf\Omega=\dot{\mathbf Q}\mathbf Q^{\sf T}=-\mathbf\Omega^{\sf T}$. The constituent velocity gradients transform as
\begin{equation}
 \Lgrad_\alpha^*=\mathbf Q\Lgrad_\alpha\mathbf Q^{\sf T}+\mathbf\Omega,
 \qquad \alpha\in\{s,p\}.
 \label{eq:L-transform}
\end{equation}
Since the coefficients in \eqref{eq:Ld-alpha} sum to one,
\begin{equation}
 \Lgrad_d^*=\mathbf Q\Lgrad_d\mathbf Q^{\sf T}+\mathbf\Omega.
 \label{eq:Ld-transform}
\end{equation}
For an objective spatial conformation tensor $\C^*=\mathbf Q\C\mathbf Q^{\sf T}$, its polymer material derivative obeys
\begin{equation}
 \Dp^*\C^*=\mathbf Q(\Dp\C)\mathbf Q^{\sf T}
 +\mathbf\Omega\C^*-\C^*\mathbf\Omega.
 \label{eq:DpC-transform}
\end{equation}
Consequently,
\begin{equation}
 \Dp^*\C^*-\Lgrad_d^*\C^*-\C^*(\Lgrad_d^*)^{\sf T}
 =\mathbf Q\left(\Dp\C-\Lgrad_d\C-\C\Lgrad_d^{\sf T}\right)\mathbf Q^{\sf T}.
 \label{eq:twofluid-objectivity}
\end{equation}
Thus the two-velocity upper-convected operator derived below is objective. The same calculation applies to the Gordon--Schowalter family in \Cref{sec:objective-rates}: its spin transforms inhomogeneously and its rate tensor homogeneously. Objectivity therefore verifies admissible observer covariance but does not select either the material transport velocity or the deformation velocity; those choices come from the population kinematics and the underlying microstructural model.

\section{Configurational free energy, Kramers stress and connector-space dissipation}
\label{sec:free-energy}

The basic ingredients of this section---the configurational free-energy functional, Kramers stress and Fokker--Planck/connector-space dissipation---are standard in dumbbell kinetic theory; see, e.g., Bird et al.~\cite{Bird1987} and \"Ottinger~\cite{Ottinger1996}. They are repeated here because the later Class-II entropy exploitation requires their mechanical-power and sign conventions in the present two-velocity notation. The relative-entropy rewriting in \Cref{sec:config-relative} is likewise a standard free-energy/relative-entropy identity; its role here is to expose the signed mechanical exchange term that will be paired with the constituent stress power, not to introduce a second entropy law.

Let $U(\bm r,T)$ be the spring free-energy potential of one dumbbell. The local configurational Helmholtz free energy per polymer molecule is taken as
\begin{equation}
 a_c[\varphi,T]
 =\int_{\R^d}\varphi U\,\dd\bm r
 +k_B T\int_{\R^d}\varphi\ln\frac{\varphi}{\varphi_*}\,\dd\bm r,
 \label{eq:config-free-energy}
\end{equation}
where $\varphi_*>0$ is a fixed constant with units of connector-space density, used only to render the logarithm dimensionless; it is not a normalized distribution on $\R^d$. All connector integrations below assume finite displayed moments, entropy and dissipation integrals, and vanishing boundary terms. Classical integrations involving $\ln\varphi$ are first understood for smooth positive densities with sufficient decay; extension to less regular densities requires an appropriate approximation argument. The free-energy density per physical volume is
\begin{equation}
 f_c=na_c.
 \label{eq:fc}
\end{equation}
Define the connector chemical potential up to an $\bm r$-independent additive term by
\begin{equation}
 \mu_c
 =U+k_BT\ln\frac{\varphi}{\varphi_*}.
 \label{eq:muconfig}
\end{equation}
Its gradient is unique,
\begin{equation}
 \nabla_{\bm r}\mu_c
 =\nabla_{\bm r}U+k_BT\nabla_{\bm r}\ln\varphi.
 \label{eq:muconfig-grad}
\end{equation}

Assume central spring forces, $U=U(|\bm r|,T)$, so that the Kramers tensor is symmetric. Define
\begin{equation}
 \taup
 =n\int_{\R^d}\varphi\,\nabla_{\bm r}U\otimes\bm r\,\dd\bm r
 -nk_BT\I.
 \label{eq:kramers-stress}
\end{equation}
Noncentral spring models are not treated here: simply symmetrizing the Kramers tensor would generally discard rotational power and requires a separately specified angular-momentum and torque balance.

\begin{proposition}[Configurational free-energy identity]
\label{prop:free-energy-identity}
Let \eqref{eq:n-balance} and \eqref{eq:phi-balance} hold, and assume sufficient decay in connector space. Define the configurational entropy density by
\begin{equation}
 s_c=-n\left.\frac{\partial a_c}{\partial T}\right|_{\varphi}.
 \label{eq:sc-def}
\end{equation}
Then
\begin{equation}
 \partial_t f_c+\Div(f_c\bm v_p)+s_c\Dp T
 =\taup:\Lgrad_d
 +n\int_{\R^d}\Jr\cdot\nabla_{\bm r}\mu_c\,\dd\bm r.
 \label{eq:fc-balance-general}
\end{equation}
\end{proposition}

\begin{proof}
Because $f_c=na_c$ and \eqref{eq:n-balance} holds,
\[
 \partial_t f_c+\Div(f_c\bm v_p)=n\Dp a_c.
\]
The chain rule gives
\[
 n\Dp a_c=-s_c\Dp T+n\int\mu_c\Dp\varphi\,\dd\bm r,
\]
where the $\bm r$-independent normalization contribution drops out because $\int\Dp\varphi\,\dd\bm r=0$. Insert \eqref{eq:phi-balance}, integrate by parts in $\bm r$, and obtain
\[
 n\int\mu_c\Dp\varphi\,\dd\bm r
 =n\int\varphi(\Lgrad_d\bm r)\cdot\nabla_{\bm r}\mu_c\,\dd\bm r
 +n\int\Jr\cdot\nabla_{\bm r}\mu_c\,\dd\bm r.
\]
The first term equals $\taup:\Lgrad_d$. Indeed, the energetic part gives the first term of \eqref{eq:kramers-stress}, while integration by parts in the entropic part yields $-nk_BT\tr\Lgrad_d$.
\end{proof}

A connector-space dissipative closure is
\begin{equation}
 \Jr=-\mathbf M_r\,\varphi\,\nabla_{\bm r}\mu_c,
 \label{eq:Jr-closure}
\end{equation}
where the symmetric part of $\mathbf M_r$ is positive semidefinite (and, in the diagonal closures used below, $\mathbf M_r=\mathbf M_r^{\sf T}\ge0$). Define
\begin{equation}
 \zeta_r
 \eqdef-\frac{n}{T}\int\Jr\cdot\nabla_{\bm r}\mu_c\,\dd\bm r
 =\frac{n}{T}\int\varphi\,
 (\nabla_{\bm r}\mu_c)\cdot\mathbf M_r(\nabla_{\bm r}\mu_c)\,\dd\bm r
 \ge0.
 \label{eq:zetar}
\end{equation}
Then \eqref{eq:fc-balance-general} becomes
\begin{equation}
 \partial_t f_c+\Div(f_c\bm v_p)+s_c\Dp T
 =\taup:\Lgrad_d-T\zeta_r.
 \label{eq:fc-balance}
\end{equation}

The term $\taup:\Lgrad_d$ is the reversible mechanical power transferred to the polymer configuration by the deformation field. The connector relaxation \eqref{eq:zetar} is dissipative and is balance-anchored by the flux $\Jr$ in \eqref{eq:phi-balance}.

\subsection{Reduced configurational relative-free-energy identity}
\label{sec:config-relative}

Define the instantaneous Gibbs equilibrium distribution
\begin{equation}
 \varphi_{\mathrm{eq}}(\bm r;T)
 =Z(T)^{-1}\exp\left(-\frac{U(\bm r,T)}{k_BT}\right)
 \label{eq:phi-eq}
\end{equation}
with $0<Z(T)<\infty$ on the temperature range under consideration. This normalizability condition is part of the admissible spring-potential assumptions. Define the nonnegative relative-entropy functional
\begin{equation}
 \mathcal H_c
 =nk_B\int\varphi\ln\frac{\varphi}{\varphi_{\mathrm{eq}}}\,\dd\bm r\ge0.
 \label{eq:relative-H}
\end{equation}
Let $a_c^{\rm eq}(T)$ denote \eqref{eq:config-free-energy} evaluated at $\varphi_{\rm eq}$ and $f_c^{\rm eq}=na_c^{\rm eq}$. Direct substitution of \eqref{eq:phi-eq} gives the exact algebraic identity
\begin{equation}
 T\mathcal H_c=f_c-f_c^{\rm eq}.
 \label{eq:H-freeenergy}
\end{equation}
It is useful to introduce the signed relative-free-energy deficit
\begin{equation}
 \eta_c\eqdef-\mathcal H_c
 =-\frac{f_c-f_c^{\rm eq}}{T}\le0.
 \label{eq:relative-entropy-deficit}
\end{equation}
This is the negative relative Helmholtz free-energy excess divided by temperature. Under isothermal, undeformed relaxation, the per-molecule quantity $\mathcal H_c/n$ is nonincreasing along polymer trajectories. The density $\mathcal H_c$ need not decrease pointwise when $n$ changes; its integral decreases on a polymer-material volume, or on a domain with vanishing net relative-free-energy transport. Thus $-\eta_c$, rather than $\eta_c$, is the nonnegative relaxation functional. Outside that setting, $\eta_c$ is distinct from the physical configurational entropy $s_c=-n\partial_Ta_c|_\varphi$.

For local test processes with $\Dp T=0$, \eqref{eq:fc-balance} implies
\begin{equation}
 \partial_t\eta_c+\Div(\eta_c\bm v_p)
 =-\frac{1}{T}\taup:\Lgrad_d+\zeta_r.
 \label{eq:partial-config-entropy}
\end{equation}
Equation~\eqref{eq:partial-config-entropy} identifies the signed reversible mechanical contribution $-\taup:\Lgrad_d/T$ to an isothermal relative-free-energy balance. It is not a balance for the physical polymer entropy and does not define a separate constituent entropy inequality.

\section{Cancellation of reversible configurational power in the total entropy balance}
\label{sec:cancellation}

The cancellation does not require \eqref{eq:partial-config-entropy}; it occurs directly in the total entropy exploitation. Proposition~\ref{prop:convective-cancellation} has already removed the conformation-gradient part of polymer diffusion. At conformation level write
\begin{equation}
 \Dp\C=\Lgrad_d\C+\C\Lgrad_d^{\sf T}+\mathbf R_C,
 \label{eq:C-general-kinematics-early}
\end{equation}
where $\mathbf R_C$ contains relaxation or other non-affine source terms. Here ``isotropic'' means invariance under $\C\mapsto\mathbf Q\C\mathbf Q^{\sf T}$ for every orthogonal $\mathbf Q$. On the smooth positive-definite state domain this implies coaxiality of the thermodynamic response with $\C$: for any skew tensor $\mathbf A$, differentiate the invariance relation along $\mathbf Q_\varepsilon=\exp(\varepsilon\mathbf A)$ at $\varepsilon=0$ to obtain
\[
 0=\Zc:(\mathbf A\C-\C\mathbf A)
   =\tr[(\C\Zc-\Zc\C)\mathbf A].
\]
Since $\C\Zc-\Zc\C$ is skew, choosing $\mathbf A=\C\Zc-\Zc\C$ gives $0=-\|\C\Zc-\Zc\C\|^2$, and hence $\C\Zc=\Zc\C$. For such an isotropic free energy $f(\rho_s,\rho_p,T,\C)$, with $\D_d=\sym\Lgrad_d$, define
\begin{equation}
 \taup=\C\Zc+\Zc\C.
 \label{eq:general-stress-Z-early}
\end{equation}
The commutation just established is essential in the next step: $\Zc\C$ is symmetric and
\[
 \Zc:(\Lgrad_d\C+\C\Lgrad_d^{\sf T})
 =2(\Zc\C):\Lgrad_d
 =2(\Zc\C):\D_d
 =(\C\Zc+\Zc\C):\D_d
 =\taup:\D_d.
\]
Consequently,
\begin{equation}
 -\Zc:\Dp\C
 =-\taup:\D_d-\Zc:\mathbf R_C.
 \label{eq:direct-config-power}
\end{equation}
The first term is sign-indefinite under the local test-process class and cannot by itself be assigned to a dissipative mechanism. Its complementary momentum-balance power will identify it as reversible exchange in the cancellation theorem below; the second term is the candidate configurational dissipation.

For constant $\alpha$ in \eqref{eq:Ld-alpha}, use the reversible elastic stress shares
\begin{equation}
 \T_p^{\rm el}=\alpha\taup,
 \qquad
 \T_s^{\rm el}=(1-\alpha)\taup,
 \label{eq:stress-division}
\end{equation}
so that
\begin{equation}
 \T_p^{\rm el}:\D_p+\T_s^{\rm el}:\D_s
 =\taup:\D_d.
 \label{eq:stress-power-identity}
\end{equation}
In the displayed gauge the elastic interaction force is zero by \eqref{eq:elastic-interaction-gauge}. A different stress/interaction representation may be used only together with the associated flux transformation; the local entropy identity is not to be refactorized while the balance fluxes are held fixed.

\begin{theorem}[Class-II reversible configurational-power cancellation]
\label{thm:cancellation}
Assume the binary balance architecture, the dilute isotropic conformation free energy \eqref{eq:dilute-f-split}, the conformation transport \eqref{eq:C-general-kinematics-early}, the ordinary reversible gauge \eqref{eq:explicit-rev-stress}--\eqref{eq:explicit-rev-interaction}, the elastic-interaction gauge \eqref{eq:elastic-interaction-gauge}, and the matching elastic-stress decomposition \eqref{eq:stress-division}. Then all sign-indefinite configurational deformation power cancels in the total entropy identity. If the conformation-level source obeys
\begin{equation}
 T\zeta_C\eqdef-\Zc:\mathbf R_C\ge0,
 \label{eq:zetaC-assumption}
\end{equation}
then the total entropy production is
\begin{equation}
 \zeta=\zeta_{\mathrm{mix}}+\zeta_C,
 \label{eq:total-zeta-final}
\end{equation}
with $\zeta_{\rm mix}$ given by \eqref{eq:classII-entropy}. Equation~\eqref{eq:total-zeta-final} is a conformation-level statement. A distribution-level entropy exploitation uses the functional state $\varphi$ rather than $\C$; for the Hookean Gaussian invariant class the resulting connector production $\zeta_r$ coincides exactly with $\zeta_C$, as shown in \Cref{prop:gaussian-dissipation-match}.
\end{theorem}

\begin{proof}
The conformation-dependent polymer chemical-potential transport first cancels the barycentric/polymer convection difference by Proposition~\ref{prop:convective-cancellation}. The remaining conformation term in $T\zeta$ is therefore $-\Zc:\Dp\C$. Equation \eqref{eq:direct-config-power} splits it into $-\taup:\D_d-\Zc:\mathbf R_C$. The constituent momentum powers contribute $+\taup:\D_d$ by \eqref{eq:stress-power-identity}; hence the sign-indefinite reversible terms cancel before any mechanism-wise positivity requirement is imposed. The ordinary thermo-chemical terms cancel by \eqref{eq:ordinary-reversible-cancellation}. What remains is exactly \eqref{eq:classII-entropy} plus \eqref{eq:zetaC-assumption}.
\end{proof}

\begin{remark}[The dilute ``surrounding fluid = solvent'' case]
For $\alpha=0$, $\Lgrad_d=\Lgrad_s$ and, in the gauge \eqref{eq:stress-division}, the elastic stress enters the solvent partial momentum equation. Polymer centers are transported with $\bm v_p$ while the connector is deformed by the solvent gradient. This allocation is gauge dependent: Proposition~\ref{prop:gauge-transform} permits a coordinated redistribution of reversible stress and interaction force together with the associated energy and entropy fluxes. Degond--Lozinski--Owens derive bead drag from the solvent velocity in dilute non-homogeneous dumbbell models~\cite{Degond2010}.
\end{remark}

\section{Balance-anchored constitutive restrictions after reversible extraction}

After the reversible cancellation, the selected mechanisms in \eqref{eq:total-zeta-final} are all balance-anchored at the level at which their constitutive restrictions are imposed:
\begin{center}
\small
\begin{tabularx}{\textwidth}{@{}l>{\raggedright\arraybackslash}X>{\raggedright\arraybackslash}X@{}}
\toprule
Mechanism & anchored constitutive factor & balance origin\\
\midrule
heat conduction & $\bm q$ & total energy flux\\
solvent/polymer viscous deformation & $\T_s^d,\T_p^d$ & partial momentum fluxes\\
relative translation & $\bm m_p^d$ & partial momentum production\\
conformation relaxation & $\mathbf R_C$ & conformation-balance production\\
kinetic connector relaxation & $\Jr$ & connector-space population flux\\
\bottomrule
\end{tabularx}
\end{center}
At conformation level, the diagonal closures \eqref{eq:fourier}--\eqref{eq:drag-closure} together with \eqref{eq:zetaC-assumption} yield
\begin{align}
 \zeta={}&\frac{1}{T^2}(\nabla T)\cdot\mathbf K_T(\nabla T)
 +\sum_{\alpha=s,p}\frac{1}{T}
 \left(2\eta_\alpha|\dev\D_\alpha|^2
 +\kappa_\alpha^b(\Div\bm v_\alpha)^2\right)
 +\frac{\gamma}{T}|\bm w|^2
 +\zeta_C\ge0.
 \label{eq:fully-positive-zeta}
\end{align}
For a kinetic model formulated directly in $\varphi$, the last term is replaced by the connector-space production \eqref{eq:zetar} after carrying out the entropy exploitation at that functional level. No constituent-wise entropy inequality has been used. At conformation level, however, the sign restriction $-\Zc:\mathbf R_C\ge0$ is an entropy-admissibility condition, not by itself a constitutive closure for $\mathbf R_C$. A closed conformation model additionally requires an explicit objective source law for $\mathbf R_C$; the Hookean specialization below supplies such a law.

\section{Moment reduction: a two-velocity conformation tensor}
\label{sec:moment}

Define the absolute conformation tensor
\begin{equation}
 \C=\int_{\R^d}\varphi\,\bm r\otimes\bm r\,\dd\bm r.
 \label{eq:C-def}
\end{equation}
Its statistical meaning is transparent from the associated quadratic form. For every unit vector $\bm n$,
\begin{equation}
 \bm n\cdot\C\bm n
 =\int_{\R^d}\varphi\,(\bm n\cdot\bm r)^2\,\dd\bm r,
 \qquad
 \tr\C=\int_{\R^d}\varphi\,|\bm r|^2\,\dd\bm r.
 \label{eq:C-directional-meaning}
\end{equation}
Hence $\sqrt{\bm n\cdot\C\bm n}$ is the root-mean-square connector component in the direction $\bm n$, while $\tr\C$ is the mean-square connector length. If $\C\bm e_i=c_i\bm e_i$, then $\sqrt{c_i}$ is the root-mean-square connector component along the principal direction $\bm e_i$; by the Rayleigh principle, an eigenvector belonging to the largest eigenvalue is a direction of maximal mean-square projection of the connector. If the connector distribution is apolar, so that $\int\varphi\bm r\,\dd\bm r=0$, then $\C$ also coincides with the covariance tensor of $\bm r$. A fuller geometric discussion of this interpretation, including the corresponding orientation surface, is given in~\cite{BotheNiethammerPilzBrenn2022}; see also the standard kinetic-theory interpretation in~\cite{Bird1987}.

Equation~\eqref{eq:C-directional-meaning} also shows directly that this second moment is positive semidefinite. Whenever the conformation-level state space $\C>0$, $\log\det\C$ or $\C^{-1}$ is used below, the kinetic description is correspondingly restricted to nondegenerate connector distributions whose second moment is positive definite (equivalently, the distribution is not supported on a proper linear subspace).
Taking the second moment of \eqref{eq:phi-balance} yields the exact identity
\begin{equation}
 \Dp\C-\Lgrad_d\C-\C\Lgrad_d^{\sf T}
 =\int_{\R^d}\left(\Jr\otimes\bm r+\bm r\otimes\Jr\right)\dd\bm r.
 \label{eq:C-general}
\end{equation}
Define the \emph{two-velocity upper-convected operator}
\begin{equation}
 \overset{\nabla(p|d)}{\C}
 \eqdef\Dp\C-\Lgrad_d\C-\C\Lgrad_d^{\sf T}.
 \label{eq:two-velocity-upper}
\end{equation}
It follows from the population kinematics: the polymer material derivative comes from center transport with $\bm v_p$, and the deformation terms come from $\Lgrad_d$. Frame indifference verifies covariance of this operator but does not by itself select it.

\subsection{Hookean dumbbells}

The Hookean dumbbell specialization used here is classical in kinetic polymer theory~\cite{Bird1987,Ottinger1996}. Take
\begin{equation}
 U(\bm r,T)=\frac{H(T)}{2}|\bm r|^2,
 \qquad H(T)>0,
 \label{eq:hookean-U}
\end{equation}
and isotropic connector mobility
\begin{equation}
 \mathbf M_r=\frac{2}{\zeta_b(T)}\I,
 \label{eq:hookean-mobility}
\end{equation}
where $\zeta_b(T)>0$ is the bead friction coefficient. Then
\begin{equation}
 \Jr=-\frac{2H}{\zeta_b}\varphi\bm r
 -\frac{2k_BT}{\zeta_b}\nabla_{\bm r}\varphi.
 \label{eq:hookean-Jr}
\end{equation}
Substitution into \eqref{eq:C-general} gives
\begin{equation}
 \overset{\nabla(p|d)}{\C}
 =-\frac{4H}{\zeta_b}\C+\frac{4k_BT}{\zeta_b}\I.
 \label{eq:C-hookean-raw}
\end{equation}
Define
\begin{equation}
 \lambda(T)=\frac{\zeta_b(T)}{4H(T)},
 \qquad
 \C_{\rm eq}(T)=\frac{k_BT}{H(T)}\I.
 \label{eq:lambda-Ceq}
\end{equation}
Thus, in the sense of \eqref{eq:C-directional-meaning}, $\sqrt{k_BT/H}$ is the equilibrium root-mean-square connector component in every direction.
Then
\begin{equation}
 \lambda\overset{\nabla(p|d)}{\C}+\C=\C_{\rm eq}.
 \label{eq:C-hookean}
\end{equation}
The Kramers stress is
\begin{equation}
 \taup=nH\C-nk_BT\I=nH(\C-\C_{\rm eq}).
 \label{eq:hookean-tau}
\end{equation}

\begin{proposition}[Preservation of positive definiteness]
\label{prop:SPD-preservation}
Along a smooth polymer trajectory, let $\lambda>0$ and $c=k_BT/H>0$ be continuous and let $\Lgrad_d$ be locally integrable. Then the Hookean equation \eqref{eq:C-hookean} preserves $\C>0$ from positive-definite initial data.
\end{proposition}
\begin{proof}
Set $\mathbf B_d=\Lgrad_d-\I/(2\lambda)$ and let $\mathbf U(t,t_0)$ be its fundamental matrix. Along the trajectory,
\begin{equation}
 \C(t)=\mathbf U(t,t_0)\C(t_0)\mathbf U(t,t_0)^{\sf T}
 +\int_{t_0}^{t}\frac{c(\sigma)}{\lambda(\sigma)}
 \mathbf U(t,\sigma)\mathbf U(t,\sigma)^{\sf T}\,\dd\sigma.
 \label{eq:SPD-variation-constants}
\end{equation}
The first term is positive definite and the integral is positive semidefinite. This proves the assertion for every time for which the stated trajectory and coefficients exist.
\end{proof}

\subsection{Conformation free energy and stress identity}

For a Gaussian closure the relative conformation free energy can be written directly as a function of $\C$,
\begin{equation}
 f_C(n,T,\C)
 =\frac{nk_BT}{2}
 \left[
 \tr\left(\frac{H\C}{k_BT}\right)
 -\ln\det\left(\frac{H\C}{k_BT}\right)-d
 \right].
 \label{eq:C-free-energy}
\end{equation}
Its derivative is
\begin{equation}
 \frac{\partial f_C}{\partial\C}
 =\frac{nH}{2}\I-\frac{nk_BT}{2}\C^{-1}.
 \label{eq:C-free-energy-derivative}
\end{equation}
Hence
\begin{equation}
 \C\frac{\partial f_C}{\partial\C}
 +\frac{\partial f_C}{\partial\C}\C
 =nH\C-nk_BT\I=\taup.
 \label{eq:stress-freeenergy-identity}
\end{equation}
This is the conformation-level Kramers stress relation. Combined with the upper-convected transport, it yields the reversible stress power used in the entropy balance. The equilibrium configurational free energy, which is independent of $\C$, is understood to be absorbed into the background part $f_0$; \eqref{eq:C-free-energy} is the nonequilibrium excess.

\begin{proposition}[Match of kinetic and conformation dissipation for Hookean Gaussians]
\label{prop:gaussian-dissipation-match}
Let $\varphi$ be the centered Gaussian with covariance $\C$ and let the Hookean connector-flux law \eqref{eq:hookean-Jr} hold. Then the Gaussian class is invariant under the linear connector Fokker--Planck dynamics, and the conformation-level dissipation equals the connector-space dissipation:
\begin{equation}
 T\zeta_C=-\Zc:\mathbf R_C
 =T\zeta_r
 =\frac{nk_BT}{2\lambda}
 \tr\left(\mathbf A+\mathbf A^{-1}-2\I\right)\ge0,
 \qquad \mathbf A=\frac{H\C}{k_BT}.
 \label{eq:gaussian-dissipation-match}
\end{equation}
\end{proposition}
\begin{proof}
For \eqref{eq:C-free-energy}, $\Zc=(nH/2)\I-(nk_BT/2)\C^{-1}$ and \eqref{eq:C-hookean} gives $\mathbf R_C=-\lambda^{-1}(\C-\C_{\rm eq})$. Direct contraction yields the last expression in \eqref{eq:gaussian-dissipation-match}. For a centered Gaussian, $\nabla_{\bm r}\ln\varphi=-\C^{-1}\bm r$; inserting this into \eqref{eq:zetar} with $\mathbf M_r=2\I/\zeta_b$ gives the same trace expression after using $\lambda=\zeta_b/(4H)$.
\end{proof}

\begin{proposition}[Free-energy and dissipation remainders beyond Gaussian closure]
\label{prop:nongaussian-remainder}
Let $\varphi$ be a smooth positive normalized density with positive-definite second moment $\C$, and assume the Hookean law and the integrability and decay needed below. Let
\[
 g_{\C}(\bm r)=\frac{\exp[-\bm r\cdot\C^{-1}\bm r/2]}{(2\pi)^{d/2}(\det\C)^{1/2}}
\]
be the centered Gaussian with the same second moment. Then
\begin{align}
 f_c-f_c^{\rm eq}
 &= f_C+nk_BT\int\varphi\ln\frac{\varphi}{g_{\C}}\,\dd\bm r,
 \label{eq:freeenergy-nongaussian-remainder}\\
 T\zeta_r
 &=T\zeta_C+\frac{2n(k_BT)^2}{\zeta_b}
 \int\varphi\left|\nabla_{\bm r}\ln\frac{\varphi}{g_{\C}}\right|^2\,\dd\bm r.
 \label{eq:dissipation-nongaussian-remainder}
\end{align}
In particular, $f_c-f_c^{\rm eq}\ge f_C$ and $\zeta_r\ge\zeta_C$, with equality precisely for $\varphi=g_{\C}$ under these assumptions.
\end{proposition}
\begin{proof}
The logarithm $\ln(g_{\C}/\varphi_{\rm eq})$ is a constant plus a quadratic form in $\bm r$. Its averages under $\varphi$ and $g_{\C}$ agree because their second moments agree. Splitting $\ln(\varphi/\varphi_{\rm eq})$ through $g_{\C}$ proves \eqref{eq:freeenergy-nongaussian-remainder} and its relative-entropy inequality. For dissipation, write
\[
 \nabla_{\bm r}\mu_c
 =(H\I-k_BT\C^{-1})\bm r
 +k_BT\nabla_{\bm r}\ln(\varphi/g_{\C}).
\]
The cross term in the squared norm vanishes because integration by parts gives
\[
 \int\varphi\,\bm r\otimes\nabla_{\bm r}\ln(\varphi/g_{\C})\,\dd\bm r
 =-\I+\C\C^{-1}=\mathbf0.
\]
The square of the linear term gives \eqref{eq:gaussian-dissipation-match}; the remaining square gives \eqref{eq:dissipation-nongaussian-remainder}. Vanishing relative Fisher information forces the ratio $\varphi/g_{\C}$ to be constant, and normalization fixes that constant to one.
\end{proof}

The distribution-level and conformation-level descriptions are alternative levels of description, not simultaneous additions to the state. In particular, $f_c[\varphi]$ and $f_C(\C)$ must not both be included in the same Helmholtz free energy. For Hookean springs, the second-moment evolution \eqref{eq:C-hookean} is already an exact closed moment equation for every sufficiently regular distribution with finite second moment. Proposition~\ref{prop:gaussian-dissipation-match} makes the stronger thermodynamic identification of $f_C$ and the relaxation production with the kinetic description on the centered Gaussian Hookean invariant class. Outside that class, Proposition~\ref{prop:nongaussian-remainder} quantifies the unresolved free energy and dissipation rather than merely withholding an equality claim.

\section{Temperature form of the total energy equation}
\label{sec:temperature-equation}

Writing the internal-energy balance \eqref{eq:internal-energy-exact} as an evolution equation for temperature retains configurational energy storage, reversible thermoelasticity, compressibility, relative constituent motion and configurational relaxation. Replacing these terms by the total stress power would in general count reversible elastic storage as heat.

\subsection{Exact temperature identity}

In this subsection, $p_T$ denotes $\partial p/\partial T$ at fixed $(\rho_s,\rho_p,\C)$ unless another set of held variables is stated explicitly; after the fixed-composition one-fluid reduction below, $p_T$ instead denotes the derivative of $p(\rho,T)$ at fixed $\rho$.

Define the volumetric heat capacity at fixed partial densities and conformation by
\begin{equation}
 \CV
 \eqdef -T\left.\frac{\partial^2 f}{\partial T^2}\right|_{\rho_s,\rho_p,\C}
 =T\left.\frac{\partial s}{\partial T}\right|_{\rho_s,\rho_p,\C}.
 \label{eq:CV-def}
\end{equation}
When \eqref{eq:temperature-twofluid} below is solved as an evolution equation for $T$, the admissible state range is restricted to $\CV>0$; the balance identity itself does not require division by $\CV$.
The thermodynamic identities \eqref{eq:thermo-conjugates} imply
\begin{equation}
 \dd u
 =\CV\,\dd T
 +\sum_{\alpha=s,p}h_\alpha\,\dd\rho_\alpha
 +(\Zc-T\Zc_T):\dd\C,
 \label{eq:u-temperature-differential}
\end{equation}
where a subscript $T$ on $\Zc$ denotes the derivative at fixed $(\rho_s,\rho_p,\C)$. Thus $\Zc-T\Zc_T$ is the conformation derivative of the \emph{internal} energy, whereas $\Zc$ is the conformation derivative of the Helmholtz free energy.

Insertion of \eqref{eq:u-temperature-differential} into the exact internal-energy balance and use of the partial mass balances yields the first temperature form
\begin{align}
 \CV\Db T={}&-\Div\bm q
 +\sum_{\alpha=s,p}\T_\alpha:\D_\alpha
 -\bm m_p\cdot\bm w
 +(p-Tp_T)\Div\bm v \notag\\
 &-\sum_{\alpha=s,p}\bm J_\alpha\cdot\nabla h_\alpha
 -(\Zc-T\Zc_T):\Db\C.
 \label{eq:temperature-raw}
\end{align}
Equation \eqref{eq:temperature-raw} is simply the total energy balance expressed in the variable $T$; no entropy-production closure has yet been inserted.

Define the conformation relaxation/source tensor
\begin{equation}
 \mathbf R_C
 \eqdef \overset{\nabla(p|d)}{\C}
 =\Dp\C-\Lgrad_d\C-\C\Lgrad_d^{\sf T}.
 \label{eq:RC-def}
\end{equation}
Because the conformation is spatially transported with the polymer velocity,
\begin{equation}
 \Db\C
 =\Lgrad_d\C+\C\Lgrad_d^{\sf T}+\mathbf R_C
 -(\bm v_p-\bm v)\cdot\nabla\C.
 \label{eq:C-bary-kinematics}
\end{equation}
For the dilute split \eqref{eq:dilute-f-split}, the conformation dependence of the polymer partial entropy satisfies
\begin{equation}
 \nabla\bar s_p
 =\nabla^{\circ}\bar s_p
 -\frac{1}{\rho_p}\,\Zc_T:\nabla\C,
 \qquad
 \nabla\bar s_s=\nabla^{\circ}\bar s_s,
 \label{eq:sbar-circle}
\end{equation}
where $\nabla^{\circ}$ means that the explicit conformation dependence is held fixed. Since
\begin{equation}
 \bm v_p-\bm v=\frac{\bm J_p}{\rho_p},
 \label{eq:up-Jp}
\end{equation}
the conformation-gradient term in \eqref{eq:C-bary-kinematics} cancels exactly with the corresponding contribution carried by the polymer relative-energy flux. This is the thermal analogue of the cancellation in \Cref{thm:cancellation}.

Define the residual mechanical dissipation density
\begin{equation}
 \Phi_d
 \eqdef
 \sum_{\alpha=s,p}\T_\alpha^d:\D_\alpha
 -\bm m_p^d\cdot\bm w.
 \label{eq:Phi-d}
\end{equation}
For the diagonal closures \eqref{eq:partial-visc}--\eqref{eq:drag-closure}, $\Phi_d\ge0$. The corresponding two-fluid conformation model then satisfies the following temperature identity once its conformation source law is specified.

\begin{theorem}[Temperature equation for the binary conformation mixture]
\label{thm:temperature}
Assume the total mixture balances, the dilute free-energy structure \eqref{eq:dilute-f-split}, the conformation transport \eqref{eq:RC-def}, the stress/transport compatibility of \Cref{thm:cancellation}, and an explicitly specified conformation-source law $\mathbf R_C$. Then
\begin{equation}
 \begin{aligned}
 \CV\Db T={}&-\Div\bm q
 +\Phi_d
 -T p_T\Div\bm v
 -T\sum_{\alpha=s,p}\bm J_\alpha\cdot\nabla^{\circ}\bar s_\alpha\\
 &+T\left(\frac{\partial\taup}{\partial T}\right)_{\rho_s,\rho_p,\C}:\D_d
 -(\Zc-T\Zc_T):\mathbf R_C.
 \end{aligned}
 \label{eq:temperature-twofluid}
\end{equation}
All thermodynamic derivatives are taken before any incompressible or isothermal reduction.
\end{theorem}

\begin{proof}
Use $h_\alpha=\mu_\alpha+T\bar s_\alpha$ in \eqref{eq:temperature-raw}. The combination
\[
 \sum_\alpha\T_\alpha:\D_\alpha-\bm m_p\cdot\bm w
 +p\Div\bm v
 -\sum_\alpha\bm J_\alpha\cdot(\nabla\mu_\alpha+\bar s_\alpha\nabla T)
 -\Zc:\Db\C
\]
is precisely the mechanical/chemical/conformation part of $T\zeta$ in \eqref{eq:exact-entropy-identity}. Insert the reversible extraction and cancellation theorem, so that the surviving mechanical part is $\Phi_d-\Zc:\mathbf R_C$. The remaining $T\Zc_T:\Db\C$ is evaluated with \eqref{eq:C-bary-kinematics}. Equation \eqref{eq:sbar-circle} cancels the explicit spatial-conformation transport. Finally, the stress identity
\begin{equation}
 \taup=\C\Zc+\Zc\C
 \label{eq:general-stress-Z}
\end{equation}
for the isotropic conformation free energy gives $\C\Zc=\Zc\C$. Differentiating this commutation relation with respect to $T$ at fixed $(\rho_s,\rho_p,\C)$ yields $\C\Zc_T=\Zc_T\C$. Therefore the same coaxiality argument used in \eqref{eq:direct-config-power} gives
\[
 T\Zc_T:(\Lgrad_d\C+\C\Lgrad_d^{\sf T})
 =T(\C\Zc_T+\Zc_T\C):\D_d
 =T\left(\frac{\partial\taup}{\partial T}\right)_{\rho_s,\rho_p,\C}:\D_d.
\]
Combining $-\Zc:\mathbf R_C$ from the entropy production with $T\Zc_T:\mathbf R_C$ from the internal-energy derivative gives the last term of \eqref{eq:temperature-twofluid}.
\end{proof}

\begin{remark}[Level of description]
At the full kinetic distribution level, the free-energy identity \eqref{eq:fc-balance-general} is exact under its stated regularity and boundary assumptions. After the connector-flux law \eqref{eq:Jr-closure} is selected, the nonnegative production \eqref{eq:zetar} is an exact consequence within that kinetic constitutive model. After reduction to a finite set of moments, the condition $T\zeta_C=-\Zc:\mathbf R_C\ge0$ is instead a thermodynamic admissibility requirement on the chosen conformation source law; it does not determine $\mathbf R_C$ by itself. For Hookean springs the second-moment evolution \eqref{eq:C-hookean} closes exactly without a Gaussian assumption, whereas equality of the finite-dimensional conformation dissipation with the connector-space production is proved in Proposition~\ref{prop:gaussian-dissipation-match} on the centered Gaussian invariant class. The temperature identity itself requires only an explicitly specified conformation source law within the stated state space and regularity assumptions. Thermodynamic admissibility is the additional restriction $-\Zc:\mathbf R_C\ge0$; when imposed, it pairs the chosen relaxation law with the chosen conformation free energy.
\end{remark}

\subsection{Term-by-term structure of the temperature equation}

Equation \eqref{eq:temperature-twofluid} contains the following contributions:
\begin{enumerate}[label=(\roman*)]
\item $-\Div\bm q$: conductive transport; $\bm q=-\mathbf K_T\nabla T$ may be anisotropic because polymer orientation can make $\mathbf K_T$ conformation dependent;
\item $\sum\T_\alpha^d:\D_\alpha$: viscous conversion of mechanical energy into heat;
\item $-\bm m_p^d\cdot\bm w$: heating by irreversible polymer--solvent relative motion (drag);
\item $-Tp_T\Div\bm v$: reversible thermo-compressive heating/cooling;
\item $-T\sum\bm J_\alpha\cdot\nabla^{\circ}\bar s_\alpha$: energy carried by relative constituent transport; in reduced one-fluid models this term is absent by construction;
\item $T(\taup)_T:\D_d$: reversible thermoelastic heating/cooling due to the temperature dependence of configurational free-energy storage;
\item $-(\Zc-T\Zc_T):\mathbf R_C$: conversion of configurational \emph{internal} energy during structural relaxation.
\end{enumerate}
Under the diagonal closure, items (ii) and (iii) are nonnegative volumetric conversion terms. Fourier conduction in item (i) is also irreversible, although it enters the temperature equation through the transport term $-\Div\bm q$ rather than as a sign-definite volumetric source. Items (iv) and (vi) are reversible. For configurational relaxation, define $Q_C^{\rm rel}=-(\Zc-T\Zc_T):\mathbf R_C$. Then
\[
 Q_C^{\rm rel}-T\zeta_C=T\Zc_T:\mathbf R_C,
 \qquad T\zeta_C=-\Zc:\mathbf R_C\ge0.
\]
Thus the relaxation heat and $T\zeta_C$ coincide for a given process when $\Zc_T:\mathbf R_C=0$; $\Zc_T=0$ is sufficient but not necessary. The local relaxation heat may have either sign even when $\zeta_C\ge0$. Entropy production and local heat release are therefore distinct quantities.

\subsection{Pressure--temperature form after one-fluid reduction}

The full binary equation is written in the density form \eqref{eq:temperature-twofluid}; a single $c_p$ is not generally appropriate while composition and relative motion remain independent. For comparison with engineering energy equations, consider the later one-velocity, fixed-composition reduction, suppose the conformation does not contribute to the thermodynamic pressure, and write $p=p(\rho,T)$. Assume the local mechanical stability condition $p_\rho>0$, which is also needed for the pressure--temperature change of variables below. Define
\begin{equation}
 \alpha_T=\frac{p_T}{\rho p_\rho},
 \qquad
 c_v=\frac{\CV}{\rho},
 \qquad
 c_p-c_v=\frac{T p_T^2}{\rho^2 p_\rho}.
 \label{eq:cp-def}
\end{equation}
Using $\dot p=-\rho p_\rho\Div\bm v+p_T\dot T$, equation \eqref{eq:temperature-twofluid} becomes
\begin{equation}
 \rho c_p\dot T
 =-\Div\bm q+\Phi_d+\alpha_TT\dot p
 +T(\taup)_T:\D
 -(\Zc-T\Zc_T):\mathbf R_C.
 \label{eq:temperature-cp}
\end{equation}
Here $c_v$ and $c_p$ are frozen-conformation, fixed-composition heat capacities, not necessarily their equilibrium values. Equation~\eqref{eq:temperature-cp} is the pressure--temperature form under these stated restrictions; it is not applicable without extra terms when $p$ depends explicitly on conformation.

\subsection{Strict Hookean thermal terms and temperature-dependent coordinates}
\label{sec:hookean-thermal-coordinates}
For the strict Hookean free energy \eqref{eq:C-free-energy}, derivatives must first be taken at fixed absolute conformation $\C$. Writing $H'=\dd H/\dd T$ gives
\begin{align}
 \Zc-T\Zc_T&=\frac n2(H-TH')\I,
 & (\taup)_T\big|_{\C}&=nH'\C-nk_B\I.
 \label{eq:strict-hookean-thermal-forces}
\end{align}
Consequently the configurational part of \eqref{eq:temperature-twofluid} is
\begin{equation}
 Q_{\rm ve}^{\rm H}
 =nTH'\C:\D_d-nk_BT\tr\D_d
 +\frac{n(H-TH')}{2\lambda}\tr(\C-\C_{\rm eq}).
 \label{eq:strict-hookean-thermal-source}
\end{equation}
Even for temperature-independent $H$, the relaxation heat is proportional to $\tr(\C-\C_{\rm eq})$ and can be negative for a contracted configuration; the relaxation entropy production remains nonnegative. For $H\propto T$, the internal-energy relaxation term vanishes and the deformation term becomes $\taup:\D_d$.

The distinction is not removed by normalizing the tensor. Set $\mathbf A=H\C/(k_BT)$ and let $F(\rho_s,\rho_p,T,\mathbf A)$ denote the same Helmholtz function rewritten in these coordinates. With $\ell(T)=\partial_T\ln(H/T)$, the chain rule gives the physical entropy and internal energy as
\begin{align}
 s&=-F_T\big|_{\mathbf A}-\ell F_{\mathbf A}:\mathbf A,
 \label{eq:entropy-coordinate-chain}\\
 u&=F-TF_T\big|_{\mathbf A}-T\ell F_{\mathbf A}:\mathbf A.
 \label{eq:energy-coordinate-chain}
\end{align}
Thus replacing $s$ by $-F_T|_{\mathbf A}$ generally changes the thermodynamic model; it is not merely a change of notation. The kinematic correction in \eqref{eq:A-thermal} and the thermodynamic corrections \eqref{eq:entropy-coordinate-chain}--\eqref{eq:energy-coordinate-chain} disappear identically for arbitrary thermal histories when $H/T$ is constant. Vanishing of the kinematic correction along an isothermal history does not, in general, remove the thermodynamic corrections or identify the two non-isothermal models.

For reference, the excess free energy \eqref{eq:C-free-energy} contributes
\begin{equation}
 \CV^C=-\frac{nk_B}{2}\left[
 \frac{T^2H''}{H}\tr(\mathbf A-\I)
 +d\left(\frac{TH'}{H}-1\right)^2\right]
 \label{eq:strict-hookean-CV}
\end{equation}
to the heat capacity at fixed $\C$. The equilibrium configurational contribution absorbed into $f_0$ must also be retained: only the total heat capacity has the required physical meaning and is restricted to be positive. In particular, a negative excess contribution does not by itself establish thermal instability.

\subsection{One-mode Oldroyd-B specialization}
\label{sec:temperature-oldroydb}

The following specialization assumes the one-velocity, fixed-composition reduction. The dimensionless tensor $\mathbf A$ is taken as a primary internal variable of a phenomenological Oldroyd-B model and is not, for a general $H(T)$, thermodynamically interchangeable with the temperature-dependent normalization $H\C/(k_BT)$; see \Cref{sec:hookean-thermal-coordinates}. The temperature-normalized Hookean tensor carries the standard Oldroyd relaxation law without an additional thermal-rate term when $\Dp\ln(H/T)=0$ along the process; this holds identically for arbitrary thermal histories if $H/T$ is constant.

For a compressible fixed-composition model, write the modal free-energy density as
\begin{equation}
 f=f_0(\rho,T)+\frac{G(\rho,T)}{2}\,g(\mathbf A),
 \qquad
 g(\mathbf A)=\tr\mathbf A-\ln\det\mathbf A-d,
 \label{eq:oldroyd-freeenergy-T}
\end{equation}
with $\mathbf A=\mathbf A^{\sf T}>0$, $G(\rho,T)>0$ and $\lambda(\rho,T)>0$ on the admissible state range. These conditions make $g(\mathbf A)\ge0$ and the relaxation dissipation below nonnegative.
The binary dilute structure inherited from \eqref{eq:dilute-f-split} corresponds, at fixed polymer mass fraction, to $G(\rho,T)=\rho\widehat G(T)$. This branch has the density homogeneity of the dilute binary theory. Strict Hookean thermodynamic equivalence further requires $G=nk_BT$ and a temperature-independent normalization, $H/T=\mathrm{const}$, as specified in \Cref{sec:hookean-thermal-coordinates}. A generic $G(\rho,T)$ defines a broader phenomenological one-fluid extension for energetic elasticity and does not inherit the binary ``no configurational pressure'' property automatically. In that broader extension the thermodynamic pressure must be the \emph{total} pressure, including the conformation contribution
\begin{equation}
 p_C=\frac{\rho G_\rho-G}{2}\,g(\mathbf A).
 \label{eq:oldroyd-pressure}
\end{equation}
Thus $p_C$ vanishes identically as a function of $\mathbf A$ when $\rho G_\rho=G$, i.e. when the modal free-energy density is homogeneous of degree one in $\rho$. At the equilibrium state $\mathbf A=\I$, $p_C$ also vanishes because $g(\I)=0$. Consequently, every compressible equation below with the dilute binary no-configurational-pressure property is to be read on the branch $G=\rho\widehat G(T)$. This homogeneity condition alone does not establish kinetic equivalence. If a generic $G(\rho,T)$ is used instead, $p_C$ must remain in the pressure and thermo-compressive terms until an independent incompressible/fixed-density reduction is imposed. In the incompressible specialization $\rho$ is fixed and $G(\rho_0,T)$ is abbreviated by $G(T)$. For this phenomenological state choice,
\begin{align}
 \mathbf Z_A&=\frac{G}{2}(\I-\mathbf A^{-1}),
 &\taup&=G(\mathbf A-\I),
 \label{eq:oldroyd-Z-tau}\\
 \mathbf R_A&=-\frac1\lambda(\mathbf A-\I).
 \label{eq:oldroyd-RA}
\end{align}
The conformation contribution to the volumetric heat capacity is
\begin{equation}
 \CV^C=-\frac{T}{2}G_{TT}(\rho,T)g(\mathbf A).
 \label{eq:CV-conformation}
\end{equation}
Thermodynamic stability requires the \emph{total} heat capacity in the admissible state range to remain positive,
\begin{equation}
 \CV=\CV^0+\CV^C>0.
 \label{eq:CV-stability}
\end{equation}
This condition may restrict strongly temperature-dependent moduli. For example, if at fixed density $G\propto T^\gamma$, then $\CV^C=-\gamma(\gamma-1)Gg(\mathbf A)/(2T)$; for $\gamma>1$ the configurational contribution is negative and must be compensated by the equilibrium heat capacity.
Furthermore,
\begin{align}
 T(\taup)_T:\D
 &=T G_T(\rho,T)(\mathbf A-\I):\D,
 \label{eq:oldroyd-thermoelastic}\\
 -(\mathbf Z_A-T(\mathbf Z_A)_T):\mathbf R_A
 &=\frac{G-TG_T}{2\lambda}
 \tr(\mathbf A+\mathbf A^{-1}-2\I).
 \label{eq:oldroyd-relaxheat}
\end{align}
Hence the reduced one-fluid temperature equation contains the explicit viscoelastic part
\begin{equation}
 Q_{\rm ve}
 =T G_T(\rho,T)(\mathbf A-\I):\D
 +\frac{G-TG_T}{2\lambda}
 \tr(\mathbf A+\mathbf A^{-1}-2\I).
 \label{eq:Qve}
\end{equation}
Introduce the logarithmic thermoelasticity parameter
\begin{equation}
 \chi(\rho,T)=\frac{T G_T(\rho,T)}{G(\rho,T)}
 \label{eq:chi}
\end{equation}
and the nonnegative relaxation dissipation
\begin{equation}
 D_{\rm rel}=\frac{G}{2\lambda}
 \tr(\mathbf A+\mathbf A^{-1}-2\I)\ge0.
 \label{eq:Drel}
\end{equation}
Then
\begin{equation}
 Q_{\rm ve}=\chi\,\taup:\D+(1-\chi)D_{\rm rel}.
 \label{eq:Qve-partition}
\end{equation}
Equation~\eqref{eq:Qve-partition} separates the entropy-elastic and energy-elastic limits. If $G\propto T$, then $\chi=1$: configurational relaxation does not directly release internal energy, while deformation produces reversible thermoelastic heating/cooling. If $G$ is temperature independent, then $\chi=0$: deformation stores internal energy and relaxation converts the corresponding free-energy loss into heat. If $\chi>1$, pure relaxation at $\D=0$ gives $Q_{\rm ve}=(1-\chi)D_{\rm rel}<0$: relaxation is still entropy producing, but the configurational subsystem absorbs thermal energy locally. Thus dissipation and local heating cannot be identified term by term.

For the incompressible one-fluid Oldroyd-B specialization, fix $\rho=\rho_0$ and write $G(T)=G(\rho_0,T)$. With $\lambda=\nu_1/G$, Fourier heat flux and solvent viscosity $\nu$, \eqref{eq:temperature-twofluid} reduces to
\begin{align}
 \left[\CV^0-\frac{T}{2}G''g(\mathbf A)\right]\dot T
 ={}&2\nu\D:\D+\Div(\kappa\nabla T)
 +TG'(\mathbf A-\I):\D \notag\\
 &+\frac{G}{2\nu_1}(G-TG')
 \tr(\mathbf A+\mathbf A^{-1}-2\I).
 \label{eq:Hron-limit}
\end{align}
After the identifications $G=\mu$, $\lambda=\nu_1/\mu$ and $\mathbf A=\mathbf B_{\kappa_{p(t)}}$, this is the explicit temperature equation derived by Hron et al.~\cite{Hron2017} (Eqs.~(4.37) and (4.42d) in their arXiv version~3), up to their incompressible deviatoric notation. In particular, the conformation-dependent heat-capacity term, reversible $TG'$ term and $(G-TG')$ relaxation term agree coefficient by coefficient. Thus the one-fluid incompressible specialization coincides with the cited equation after identification of notation.

Wapperom--Hulsen~\cite{WapperomHulsen1998} formulate the viscoelastic heat source using an energy-partition coefficient between elastic stress work and mechanical relaxation and analyze power-law temperature/density scaling of the modulus. In this notation, if $G\propto T^\gamma\rho^\delta$, then $\chi=\gamma$ at fixed density, and for the density-homogeneous case $\delta=1$ equation~\eqref{eq:Qve-partition} has precisely that stress-work/relaxation partition structure. For $\delta\ne1$, \eqref{eq:oldroyd-pressure} shows explicitly that conformation contributes to pressure/thermal expansion; the corresponding correction is carried in the general $-Tp_T\Div\bm v$ term before any pressure reduction that assumes $p_C=0$. The corresponding partition coefficient is therefore determined by the density and temperature dependence of the stored free energy rather than introduced as an additional heat-source parameter.

\begin{remark}[Relation to the strict Hookean dumbbell]
For the kinetic Hookean model of \Cref{sec:moment}, the dimensionless tensor $\mathbf A=H\C/(k_BT)$ has stress modulus $G=nk_BT$. At fixed composition $n\propto\rho$, so the stress scale has $G\propto\rho T$ and the configurational pressure vanishes. However, inserting the formal value $\chi=1$ into \eqref{eq:Qve-partition} recovers the strict Hookean thermal source only when the temperature-dependent coordinate corrections vanish identically, in particular for $H\propto T$. For a general $H(T)$, use \eqref{eq:strict-hookean-thermal-source}, not the phenomenological partition based on $-F_T|_{\mathbf A}$. A general $G(\rho,T)$ with a primary $\mathbf A$ and the standard relaxation law is a broader internal-variable model.
\end{remark}

\subsection{Comparison with established non-isothermal viscoelastic energy equations}
\label{sec:temperature-literature}

Non-isothermal viscoelastic energy equations have been derived in several thermodynamic and kinetic formulations. Equation~\eqref{eq:temperature-twofluid} retains two binary-mixture terms absent from one-fluid theories: drag heating and entropy/relative-energy transport by relative constituent motion.

{\small
\begin{longtable}{@{}>{\raggedright\arraybackslash}p{0.23\textwidth}>{\raggedright\arraybackslash}p{\dimexpr0.77\textwidth-2\tabcolsep\relax}@{}}
\toprule
Reference & Relation to \eqref{eq:temperature-twofluid}\\
\midrule
\endfirsthead
\toprule
Reference & Relation to \eqref{eq:temperature-twofluid} (continued)\\
\midrule
\endhead
\midrule
\multicolumn{2}{r}{\emph{Continued on the next page}}\\
\endfoot
\bottomrule
\endlastfoot
Marrucci (1972)~\cite{Marrucci1972} & Derives non-isothermal free energy and stress from a dumbbell model and separates polymer free-energy storage from heat production.\\
Peters--Baaijens (1997)~\cite{PetersBaaijens1997} & Shows that the same slip/deformation kinematics entering the constitutive equation also enter the energy equation, and separates dissipated from elastically stored energy and entropy from energy elasticity.\\
Wapperom--Hulsen (1998)~\cite{WapperomHulsen1998} & Gives a compressible internal-variable temperature equation with nonequilibrium heat capacity, thermal expansion, anisotropic heat conduction, reversible elastic heating and dissipation. Their formulation uses an energy-partition coefficient between stress work and mechanical relaxation and analyzes power-law temperature/density scaling of elastic moduli; \eqref{eq:Qve-partition} is the corresponding partition in the density-homogeneous specialization.\\
Hron et al. (2017)~\cite{Hron2017} & For an incompressible Oldroyd-B/Maxwell model with $G=G(T)$, their explicit temperature equation contains the same three characteristic terms with the same coefficients: $-TG''g(\mathbf A)/2$ in the heat capacity, $TG'(\mathbf A-\I):\D$, and $(G-TG')\tr(\mathbf A+\mathbf A^{-1}-2\I)/(2\lambda)$. Thus the one-fluid incompressible limit agrees coefficient by coefficient after the notation is identified.\\
Guaily (2015)~\cite{Guaily2015} & Starts from $e=e(\rho,T,\C)$ and obtains a compressible conformation-dependent temperature equation directly from the energy balance. This is the one-fluid analogue of the raw identity \eqref{eq:temperature-raw}.\\
Dostal\'ik--M\'alek--Pr\r{u}\v{s}a--S\"uli (2020)~\cite{Dostalik2020} & Derives a compressible non-isothermal Fokker--Planck model and an explicit temperature equation with pressure work, Fourier conduction, viscous heating and separate entropic/energetic spring contributions. Their starting approximation neglects polymer mass and identifies mixture density/velocity with the solvent; the binary model considered here retains both constituent masses, velocities and drag.\\
Rodrigues et al. (2025)~\cite{Rodrigues2025} & Derives additional energy-equation terms for temperature- and time-dependent material properties and implements them computationally.\\
\end{longtable}
}

\begin{remark}[One-conformation state-space limitation]
The balance and temperature identities above are exact within the chosen one-conformation state space for any explicitly specified conformation source law satisfying the stated regularity assumptions. Thermodynamic admissibility additionally requires $T\zeta_C=-\Zc:\mathbf R_C\ge0$; that inequality restricts the source law but is not needed to establish the algebraic identities themselves. Energetic structural effects below the chain end-to-end scale may require additional internal variables~\cite{Hutter2009}; experiments and simulations also indicate significant energetic-elastic contributions at strong deformation~\cite{Ionescu2008}. These effects lie outside the one-conformation state space used here.
\end{remark}

\section{Compressibility, temperature and the stress-rate equation}
\label{sec:compressible-stress}

Before the one-velocity, incompressible and isothermal reductions, the Hookean stress equation contains additional concentration, volumetric and thermal terms.

Apply the two-velocity upper-convected operator \eqref{eq:two-velocity-upper}, now acting on stress as $\overset{\nabla(p|d)}{\taup}:=\Dp\taup-\Lgrad_d\taup-\taup\Lgrad_d^{\sf T}$, to \eqref{eq:hookean-tau}. This notation keeps the polymer transport and deformation velocities explicit. For $H=H(T)$, $\zeta_b=\zeta_b(T)$ and \eqref{eq:n-balance}, one obtains
\begin{equation}
 \overset{\nabla(p|d)}{\taup}
 +\left(\frac{1}{\lambda}+\Div\bm v_p-\Dp\ln H\right)\taup
 =2nk_BT\D_d
 +nk_BT\left(\Dp\ln H-\Dp\ln T\right)\I.
 \label{eq:compressible-nonisothermal-stress}
\end{equation}
where
\begin{equation}
 \D_d=\sym\Lgrad_d.
 \label{eq:Dd}
\end{equation}
Equation \eqref{eq:compressible-nonisothermal-stress} contains explicit concentration, compression and temperature contributions that disappear under the standard idealizations.

\begin{remark}[Isothermal specialization]
If a dimensionless conformation $\mathbf A=H(T)\C/(k_BT)$ is introduced before the thermal reduction, then
\begin{equation}
 \overset{\nabla(p|d)}{\mathbf A}
 =-\frac1\lambda(\mathbf A-\I)
 +\mathbf A\,\Dp\ln\frac{H(T)}{T}.
 \label{eq:A-thermal}
\end{equation}
The standard isothermal equation is obtained when the corresponding thermal term vanishes. The thermal term disappears identically for an ideal entropic spring with $H\propto T$, and it disappears for any $H(T)$ after the isothermal reduction.
\end{remark}

\subsection{Density-weighted conformation and a Truesdell-type rate}

Define the conformation-density tensor
\begin{equation}
 \mathbf M=n\C.
 \label{eq:Mdef}
\end{equation}
Using \eqref{eq:n-balance},
\begin{equation}
 \Dp\mathbf M-\Lgrad_d\mathbf M-\mathbf M\Lgrad_d^{\sf T}
 +(\Div\bm v_p)\mathbf M
 =n\overset{\nabla(p|d)}{\C}.
 \label{eq:truesdell-M}
\end{equation}
Thus the density-weighted state carries a Truesdell-type volumetric correction. Accordingly, the upper-convected and Truesdell-type forms may represent the same transport expressed in specific and density-weighted conformation variables.

\section{Objective-rate selection as a downstream constitutive test}
\label{sec:objective-rates}

The preceding kinetic derivation selects the affine two-velocity upper-convected operator for the chosen dumbbell transport. The entropy balance supplies an independent downstream compatibility test: changing the objective rate while retaining the same free energy, Kramers stress and reversible balance channels must not leave a sign-indefinite power. To make this test explicit, consider the phenomenological modification
\begin{equation}
 \mathbf K_a=\W_d+a\D_d,
 \qquad
 \W_d=\tfrac12(\Lgrad_d-\Lgrad_d^{\sf T}),
 \label{eq:Ka}
\end{equation}
and replace $\Lgrad_d\bm r$ in \eqref{eq:phi-balance} by $\mathbf K_a\bm r$. The second-moment rate becomes
\begin{equation}
 \overset{\diamond a}{\C}
 =\Dp\C-\W_d\C+\C\W_d-a(\D_d\C+\C\D_d).
 \label{eq:GS-C}
\end{equation}
In this convention $a=1$ is upper convected, $a=0$ is Jaumann/corotational, and $a=-1$ is lower convected.

Since $\taup$ is symmetric, the configurational free-energy power becomes
\begin{equation}
 \taup:\mathbf K_a=a\taup:\D_d.
 \label{eq:tau-Ka}
\end{equation}
Hence the reduced configurational free-energy identity carries the signed term
\begin{equation}
 -\frac{a}{T}\taup:\D_d.
 \label{eq:supply-a}
\end{equation}
If the mechanical Kramers stress transmitted through the momentum balances is kept equal to $\taup$, its mechanical power is $+\taup:\D_d/T$. The total entropy balance therefore contains the remainder
\begin{equation}
 \frac{1-a}{T}\taup:\D_d.
 \label{eq:indefinite-remainder}
\end{equation}

\begin{proposition}[Rate--stress compatibility test]
\label{prop:rate-test}
For the dumbbell/conformation state space with Kramers stress \eqref{eq:kramers-stress}, suppose the connector deformation is replaced by \eqref{eq:Ka} while the same Kramers stress is retained as the complete reversible polymer elastic stress. Assume the local test-process class is rich enough that, at a fixed non-equilibrium state with nonzero relevant component of $\taup$, the conjugate component of $\D_d$ can take both signs. If this remainder is required to vanish as an unmatched reversible power in the selected mechanism decomposition, and no additional reversible balance contribution is introduced, then admissibility within that constitutive framework requires
\begin{equation}
 a=1.
 \label{eq:a1}
\end{equation}
\end{proposition}

\begin{proof}
The uncancelled term \eqref{eq:indefinite-remainder} is linear in $\D_d$. By the stated richness hypothesis its contraction with the fixed non-equilibrium stress can be made positive or negative while the state is held fixed. It cannot therefore be retained as a dissipative mechanism. If no additional reversible balance term cancels it, its coefficient must vanish.
\end{proof}

\begin{remark}[Scope of the Gordon--Schowalter test]
The proposition excludes the combination of the same dumbbell free energy, the same Kramers stress, no additional reversible term and a non-affine rate $a\ne1$. A Jaumann rate may be associated with a different internal state, as in stress-based Maxwell thermodynamics, and a lower-convected rate arises for dual or inverse variables. The entropy principle tests a constitutive package, not the type of objective derivative in isolation. In particular, this is not a theorem that the bare inequality $\zeta\ge0$ forces $a=1$: allowing different dissipative couplings or regrouping the remainder with other terms changes the constitutive test.
\end{remark}

\begin{remark}[Modified Gordon--Schowalter constructions]
Equation \eqref{eq:indefinite-remainder} requires either a modified reversible stress or an additional reversible balance contribution. Specifically, the mechanically transmitted reversible stress may be reduced to $a\taup$, or an additional modeled reversible channel must supply exactly the complementary power $(a-1)\taup:\D_d$. An interconstituent interaction force alone cannot provide such a cancellation in states with $\bm w=\bm0$; any alternative construction must therefore identify a balance variable and power pairing capable of carrying the required term. Such extra physics can represent non-affine network slip or another structural mechanism. This qualification is consistent with the slip-tensor thermodynamic construction of Peters--Baaijens~\cite{PetersBaaijens1997}, in which changing the kinematics changes the energy equation as part of the same constitutive package. The rate parameter cannot be changed while all other constitutive structures are frozen.
\end{remark}

\subsection{Inverse conformation and lower-convected transport}

Let $\mathbf B=\C^{-1}$. Differentiating $\mathbf B\C=\I$ gives
\begin{equation}
 \Dp\mathbf B=-\mathbf B(\Dp\C)\mathbf B.
 \label{eq:inverse-derivative}
\end{equation}
If $\C$ satisfies the two-velocity upper-convected transport, then
\begin{equation}
 \Dp\mathbf B+\Lgrad_d^{\sf T}\mathbf B+\mathbf B\Lgrad_d
 =-\mathbf B\,
 \overset{\nabla(p|d)}{\C}\,\mathbf B.
 \label{eq:inverse-lower}
\end{equation}
Thus the same microstructure is upper-convected in $\C$ and lower-convected in the inverse variable $\mathbf B$; admissibility of a tensor rate cannot be assessed independently of the chosen state variable.

\section{Deformation velocity and elastic-stress decomposition as a Class-II compatibility choice}
\label{sec:stress-division}

The parameter $\alpha$ in \eqref{eq:Ld-alpha} has a mechanical counterpart: cancellation of reversible configurational power requires the same weights in the elastic-stress decomposition \eqref{eq:stress-division}. The same matching between deformation weights and elastic-stress shares is also found in two-fluid formulations based on rheological-power equality and nonlinear hydrodynamics~\cite{Tanaka1997,PleinerHarden2004}.

Under the entropy-invariant Class-II to Class-I reduction of \Cref{sec:classI-reduction}, the same parameter enters
$\beta=(\alpha-\omega_p)/M_{sp}$ and therefore the inherited configurational-stress term $-\beta\Div\taup$ in the diffusion force \eqref{eq:classI-diffusion-force}. Proposition~\ref{prop:classI-gauge-covariance} shows that the resulting diffusion force is invariant under coordinated reversible stress--interaction representation changes. Zhou--Doi likewise obtain deformation-velocity-dependent stress contributions to polymer diffusion from Onsager's principle~\cite{ZhouDoi2022}. The matching relation is therefore a compatibility condition for the specified Class-II constitutive representation, not a uniqueness statement for the elastic-stress decomposition in all two-fluid models.

\section{Entropy-invariant Class-II to Class-I reduction}
\label{sec:classI-reduction}

The Class-II parent model contains more momentum information than a Class-I mixture. The reduction follows the entropy-invariant Class-II to Class-I construction for molecular mixtures~\cite{BotheDreyer2015}. Its constitutive formulas are obtained by algebraic identification, without an asymptotic estimate. This does not make the reduction an exact change of dynamical variables: the interaction-force identification below eliminates relative acceleration, and the standard reduced momentum/energy architecture additionally omits quadratic relative inertia. Throughout this section, ``for every thermodynamic process'' means every local process satisfying the algebraic Class-I/Class-II identification conditions stated below; it does not mean that an arbitrary unconstrained Class-II solution is dynamically equivalent to a Class-I solution. The partial mass balances and the mixture internal-energy balance are identified between the two model classes, the constitutive quantities are transformed accordingly, and the entropy flux is chosen so that
\begin{equation}
 \zeta^{\rm I}\equiv\zeta^{\rm II}
 \qquad\text{for every thermodynamic process.}
 \label{eq:entropy-invariant-reduction}
\end{equation}
After the relative-momentum information has been replaced by the constitutive identification, the Class-I momentum balance differs from the exact summed momentum balance by the quadratic diffusion-inertia flux. On the identification manifold, that flux and relative kinetic-energy storage form a reversible pair. They are omitted together rather than interpreted as missing entropy production. Neither this conditional entropy identity nor the inertia bookkeeping is a convergence theorem, and both are distinct from the subsequent zero-diffusion limit.

\subsection{Algebraic change of variables}

Introduce the single independent binary diffusion flux
\begin{equation}
 \bm j\eqdef\bm J_p=-\bm J_s,
 \qquad
 M_{sp}=\frac{\rho_s\rho_p}{\rho}.
 \label{eq:classI-j-def}
\end{equation}
Then the Class-II constituent velocities are represented algebraically by the Class-I variables as
\begin{equation}
 \bm v_p=\bm v+\frac{\bm j}{\rho_p},
 \qquad
 \bm v_s=\bm v-\frac{\bm j}{\rho_s},
 \qquad
 \bm w=\frac{\bm j}{M_{sp}}.
 \label{eq:classI-velocity-identification}
\end{equation}
Equation \eqref{eq:classI-velocity-identification} is algebraic. The two partial mass balances therefore become
\begin{align}
 \partial_t\rho_p+\Div(\rho_p\bm v+\bm j)&=0,\notag\\
 \partial_t\rho_s+\Div(\rho_s\bm v-\bm j)&=0.
 \label{eq:classI-mass}
\end{align}
For later use set
\begin{align}
 \bm u_p&=\frac{\bm j}{\rho_p},
 &\bm u_s&=-\frac{\bm j}{\rho_s},
 \label{eq:classI-diffusion-velocities}\\
 \beta&\eqdef \frac{\alpha}{\rho_p}-\frac{1-\alpha}{\rho_s}
 =\frac{\alpha-\omega_p}{M_{sp}}.
 \label{eq:classI-beta}
\end{align}
The polymer transport velocity and the deformation velocity inherited from Class II are consequently
\begin{equation}
 \bm v_p=\bm v+\bm u_p,
 \qquad
 \bm v_d=\alpha\bm v_p+(1-\alpha)\bm v_s
 =\bm v+\beta\bm j.
 \label{eq:classI-vp-vd}
\end{equation}
Define the corresponding deformation tensors
\begin{align}
 \D_p^{\rm I}&=\sym\nabla\!\left(\bm v+\frac{\bm j}{\rho_p}\right),\notag\\
 \D_s^{\rm I}&=\sym\nabla\!\left(\bm v-\frac{\bm j}{\rho_s}\right),\notag\\
 \D_d^{\rm I}&=\sym\nabla(\bm v+\beta\bm j).
 \label{eq:classI-deformation-tensors}
\end{align}
The superscript ``I'' indicates that these are no longer independent constituent rates: they are algebraic differential expressions in the Class-I variables $(\rho_s,\rho_p,\bm v,\bm j)$.

The conformation equation is therefore not yet a one-velocity upper-convected law.  The exact Class-I image of \eqref{eq:C-general-kinematics-early} is
\begin{equation}
 \partial_t\C
 +\left(\bm v+\frac{\bm j}{\rho_p}\right)\!\cdot\nabla\C
 -\nabla(\bm v+\beta\bm j)\C
 -\C\nabla(\bm v+\beta\bm j)^{\sf T}
 =\mathbf R_C.
 \label{eq:classI-conformation}
\end{equation}
Polymer mass is transported with $\bm v+\bm j/\rho_p$, whereas conformation is deformed with $\bm v+\beta\bm j$. The specialization $\bm j=\bm0$ collapses both velocities to $\bm v$.

\subsection{Momentum balance and the diffusion-inertia remainder}

Summing the two Class-II momentum balances gives the exact identity
\begin{equation}
 \partial_t(\rho\bm v)
 +\Div\!\left(\rho\bm v\otimes\bm v
 +\frac{\bm j\otimes\bm j}{M_{sp}}\right)
 =\Div(\T_s+\T_p).
 \label{eq:exact-summed-momentum-classI-vars}
\end{equation}
Indeed,
\[
 \rho_s\bm v_s\otimes\bm v_s+\rho_p\bm v_p\otimes\bm v_p
 =\rho\bm v\otimes\bm v+\frac{\bm j\otimes\bm j}{M_{sp}}.
\]
Define the quadratic relative-inertia tensor and associated kinetic-energy density by
\begin{equation}
 \mathbf R_d\eqdef\frac{\bm j\otimes\bm j}{M_{sp}}
 =M_{sp}\bm w\otimes\bm w,
 \qquad
 K_d\eqdef\frac{|\bm j|^2}{2M_{sp}}
 =\frac12M_{sp}|\bm w|^2.
 \label{eq:relative-inertia-pair}
\end{equation}
The Class-II kinetic-energy density then decomposes exactly as
\begin{equation}
 \frac12\rho_s|\bm v_s|^2+\frac12\rho_p|\bm v_p|^2
 =\frac12\rho|\bm v|^2+K_d.
 \label{eq:kinetic-energy-decomposition-classI}
\end{equation}
The momentum correction $\mathbf R_d$ and the relative kinetic-energy storage $K_d$ are of the same quadratic order and form the relative-inertia pair analyzed below.

The energy-balance identification gives the Class-I stress
\begin{equation}
 \T^{\rm I}\eqdef\T_s+\T_p
 =-p\I+\taup+\T_s^d+\T_p^d,
 \label{eq:classI-total-stress}
\end{equation}
and the Class-I momentum equation is
\begin{equation}
 \partial_t(\rho\bm v)+\Div(\rho\bm v\otimes\bm v)
 =\Div\T^{\rm I}.
 \label{eq:classI-momentum}
\end{equation}
Hence \eqref{eq:classI-momentum} differs from the exact summed Class-II momentum equation only by
$\Div\mathbf R_d$, which is quadratic in the diffusion velocity.  This is the momentum-flux truncation after the independent relative-momentum dynamics has been eliminated; it is not the only loss of dynamical information in the reduction and does not set the diffusion flux to zero.

For the diagonal partial viscous closures \eqref{eq:partial-visc}, the inherited Class-I stress is explicitly
\begin{equation}
 \T^{\rm I}=-p\I+\taup
 +\T_p^d(\D_p^{\rm I})+\T_s^d(\D_s^{\rm I}).
 \label{eq:classI-stress-closure}
\end{equation}
Thus the Class-I stress generally contains gradients of the diffusion flux. This produces a higher-order PDE structure unless additional restrictions are imposed on the partial viscosities.

\subsection{Exact identification of the internal-energy balance}

Let
\begin{equation}
 \bm\Psi
 \eqdef \T_p\bm u_p+\T_s\bm u_s
 =\left(\frac{\T_p}{\rho_p}-\frac{\T_s}{\rho_s}\right)\bm j.
 \label{eq:Psi-classI}
\end{equation}
Because the ordinary reversible stresses are $-\omega_\alpha p\I$, their contribution to $\bm\Psi$ cancels identically.  Hence
\begin{equation}
 \bm\Psi
 =\left[
 \beta\taup
 +\frac{\T_p^d}{\rho_p}-\frac{\T_s^d}{\rho_s}
 \right]\bm j.
 \label{eq:Psi-nonhydrostatic}
\end{equation}
Using symmetry of the stresses,
\begin{equation}
 \sum_{\alpha=s,p}\T_\alpha:\D_\alpha
 =\T^{\rm I}:\D
 +\Div\bm\Psi
 -\sum_{\alpha=s,p}\bm u_\alpha\cdot\Div\T_\alpha.
 \label{eq:classI-stress-power-split}
\end{equation}
Substitution into the exact Class-II internal-energy balance \eqref{eq:internal-energy-exact} shows that it becomes the Class-I internal-energy balance
\begin{equation}
 \partial_t u
 +\Div\!\left(u\bm v+\bm j_u^{\rm I}\right)
 =\T^{\rm I}:\D
 \label{eq:classI-internal-energy}
\end{equation}
with
\begin{equation}
 \bm j_u^{\rm I}\eqdef\bm j_u-\bm\Psi,
 \label{eq:classI-internal-energy-flux}
\end{equation}
under the vectorial Class-II--to--Class-I interaction-force identification
\begin{equation}
 \bm m_p
 =-\Div\T_p+\omega_p\Div\T^{\rm I}.
 \label{eq:classI-energy-force-identification}
\end{equation}
Indeed, $\bm u_p=\omega_s\bm w$ and $\bm u_s=-\omega_p\bm w$ give
\[
 \sum_\alpha\bm u_\alpha\cdot\Div\T_\alpha
 =\bm w\cdot\big(\omega_s\Div\T_p-\omega_p\Div\T_s\big),
\]
so \eqref{eq:classI-energy-force-identification} makes the residual bracket
$\sum_\alpha\bm u_\alpha\cdot\Div\T_\alpha+\bm m_p\cdot\bm w$
vanish identically.  We adopt \eqref{eq:classI-energy-force-identification} as a \emph{vectorial constitutive identification} between the two model classes.

It is important not to infer more from scalar energy-power equality than it contains.  Without an additional constitutive restriction, requiring only
\[
 \bm w\cdot\Big[
 \omega_s\Div\T_p-\omega_p\Div\T_s+\bm m_p
 \Big]=0
\]
for every realized process does not by itself imply that the bracketed vector vanishes: a power-orthogonal contribution $\bm g(\bm w)$ satisfying $\bm g(\bm w)\cdot\bm w=0$ is invisible to this scalar identity.  The converse implication is valid, for example, when the admissible constitutive class excludes such power-orthogonal interaction terms (or when the interaction vector is held fixed while $\bm w$ is independently varied in the local test-process argument).  No such converse is needed below: equation~\eqref{eq:classI-energy-force-identification} is part of the stated reduction map, and all subsequent Class-I constitutive formulas are conditional on that identification.

\begin{proposition}[Dynamical meaning of the force identification]
\label{prop:relative-acceleration-identification}
On a parent Class-II balance process, define $\bm a_\alpha=(\partial_t+\bm v_\alpha\cdot\nabla)\bm v_\alpha$ and
\begin{equation}
 \bm{\mathcal A}=\omega_s\Div\T_p-\omega_p\Div\T_s+\bm m_p.
 \label{eq:relative-force-residual}
\end{equation}
Then
\begin{equation}
 \bm{\mathcal A}=M_{sp}(\bm a_p-\bm a_s).
 \label{eq:relative-acceleration-residual}
\end{equation}
Thus \eqref{eq:classI-energy-force-identification} is equivalent to $\bm a_p=\bm a_s$ on parent processes. It is a constitutive reduction condition, not an identity obeyed by arbitrary Class-II solutions.
\end{proposition}
\begin{proof}
The partial momenta give $\rho_p\bm a_p=\Div\T_p+\bm m_p$ and $\rho_s\bm a_s=\Div\T_s-\bm m_p$. Multiply these by $\omega_s$ and $\omega_p$, respectively, subtract, and use $\omega_s\rho_p=\omega_p\rho_s=M_{sp}$.
\end{proof}

\begin{remark}[A homogeneous relative-motion transient]
In nondimensional variables take spatially uniform fields with $\rho_p=\rho_s=\gamma=1$ and no external forces. The partial momenta give $\dot{\bm w}=-2\bm w$, hence $\bm w(t)=\bm w(0)e^{-2t}$, although $\Div\mathbf R_d=\bm0$ throughout. The reduced force is $\bm F^{\rm I}=\bm0$, so the algebraic Class-I law instead imposes $\bm w=\bm0$. This elementary mechanical example shows that omitting $\Div\mathbf R_d$ alone does not produce the Class-I dynamics. For the parent transient, $\dot K_d=-|\bm w|^2$ and the same power appears as positive drag conversion in the internal-energy balance.
\end{remark}

Before imposing that condition, the exact stress-power split gives
\begin{equation}
 \partial_tu+\Div(u\bm v+\bm j_u-\bm\Psi)
 =\T^{\rm I}:\D-\bm w\cdot\bm{\mathcal A}.
 \label{eq:internal-energy-before-identification}
\end{equation}
The discarded internal-energy exchange is therefore $-M_{sp}\bm w\cdot(\bm a_p-\bm a_s)$. Its smallness cannot be inferred from $\mathbf R_d=O(|\bm w|^2)$ without also controlling relative accelerations.

Since $\bm j_u=\bm q+(h_p-h_s)\bm j$, define the reduced Class-I heat flux by
\begin{equation}
 \bm j_u^{\rm I}
 =\bm q^{\rm I}+(h_p-h_s)\bm j.
 \label{eq:classI-reduced-heat-def}
\end{equation}
Equations \eqref{eq:classI-internal-energy-flux} and \eqref{eq:classI-reduced-heat-def} then give the exact constitutive identification
\begin{equation}
 \bm q^{\rm I}=\bm q-\bm\Psi.
 \label{eq:classI-heat-identification}
\end{equation}
The reduced heat flux therefore contains the inherited stress--diffusion correction $-\bm\Psi$.

\begin{proposition}[Relative-inertia consistency]
\label{prop:relative-inertia-consistency}
Assume the energy-identified Class-I stress \eqref{eq:classI-total-stress} and the exact reduced internal-energy balance \eqref{eq:classI-internal-energy}.  If one nevertheless retains the exact summed momentum flux $\mathbf R_d$ from \eqref{eq:exact-summed-momentum-classI-vars}, then the barycentric kinetic-energy balance contains the reversible power $\mathbf R_d:\D$.  Consequently the energy density $u+\frac12\rho|\bm v|^2$ has the residual $\mathbf R_d:\D$ in its natural total-energy balance. Recovering the parent conservative architecture requires the relative storage $K_d$ and associated transport, or another explicitly specified energy modification carrying the same residual.  Thus simply adding $\mathbf R_d$ to the standard Class-I momentum equation while leaving the Class-I energy architecture unchanged does not recover the exact parent balance structure. This statement concerns consistency of the reduced balance architecture after the energy identification; it does not reintroduce the discarded relative-momentum balance as an independent evolution equation for $\bm j$.
\end{proposition}

\begin{proof}
Using total mass conservation and symmetry of $\T^{\rm I}$ and $\mathbf R_d$, contraction of the exact summed momentum identity with $\bm v$ gives
\begin{equation}
 \begin{aligned}
 \partial_t\left(\frac12\rho|\bm v|^2\right)
 +\Div\left(\frac12\rho|\bm v|^2\bm v\right)
 ={}&\Div\big((\T^{\rm I}-\mathbf R_d)\bm v\big)\\
 &-\T^{\rm I}:\D+\mathbf R_d:\D.
 \end{aligned}
 \label{eq:barycentric-kinetic-with-Rd}
\end{equation}
Adding the exact Class-I internal-energy balance \eqref{eq:classI-internal-energy} yields
\begin{equation}
 \begin{aligned}
 &\partial_t\left(u+\frac12\rho|\bm v|^2\right)\\
 &\quad+\Div\left[
 \left(u+\frac12\rho|\bm v|^2\right)\bm v
 +\bm j_u^{\rm I}
 -(\T^{\rm I}-\mathbf R_d)\bm v
 \right]
 =\mathbf R_d:\D.
 \end{aligned}
 \label{eq:barycentric-total-energy-residual}
\end{equation}
The associated transport can also be displayed explicitly.  Define the relative kinetic-energy flux
\begin{equation}
 \bm Q_d\eqdef\frac12\sum_{\alpha=s,p}\rho_\alpha|\bm u_\alpha|^2\bm u_\alpha
 =\frac{|\bm j|^2}{2}
 \left(\frac{1}{\rho_p^2}-\frac{1}{\rho_s^2}\right)\bm j.
 \label{eq:relative-kinetic-flux}
\end{equation}
Then the constituent kinetic-energy flux satisfies the exact identity
\begin{equation}
 \sum_{\alpha=s,p}\frac12\rho_\alpha|\bm v_\alpha|^2\bm v_\alpha
 =\left(\frac12\rho|\bm v|^2+K_d\right)\bm v
 +\mathbf R_d\bm v+\bm Q_d.
 \label{eq:kinetic-flux-decomposition}
\end{equation}
Using the definition of the parent total-energy flux together with \eqref{eq:classI-internal-energy-flux}, the exact Class-II total-energy balance can therefore be written entirely in Class-I variables as
\begin{equation}
 \begin{aligned}
 &\partial_t\left(u+\frac12\rho|\bm v|^2+K_d\right)\\
 &\quad+\Div\left[
 \left(u+\frac12\rho|\bm v|^2+K_d\right)\bm v
 +\bm j_u^{\rm I}
 -(\T^{\rm I}-\mathbf R_d)\bm v
 +\bm Q_d
 \right]=0.
 \end{aligned}
 \label{eq:parent-total-energy-classI-vars}
\end{equation}
For an unrestricted parent Class-II process, subtraction of barycentric kinetic energy from constituent kinetic energy gives instead
\begin{equation}
 \partial_tK_d+\Div(K_d\bm v+\bm Q_d)
 =\bm w\cdot\bm{\mathcal A}-\mathbf R_d:\D.
 \label{eq:relative-kinetic-general}
\end{equation}
This equation, together with \eqref{eq:internal-energy-before-identification}, displays the equal and opposite relative-acceleration powers before any identification. The full relative dynamics is not a purely reversible subsystem; it includes drag through $\bm{\mathcal A}$. On a parent process that additionally satisfies \eqref{eq:classI-energy-force-identification}, $\bm{\mathcal A}=\mathbf0$. Equivalently, subtracting \eqref{eq:barycentric-total-energy-residual} from \eqref{eq:parent-total-energy-classI-vars} gives the exact complementary storage law
\begin{equation}
 \partial_tK_d+\Div(K_d\bm v+\bm Q_d)=-\mathbf R_d:\D.
 \label{eq:relative-kinetic-storage-law}
\end{equation}
The right-hand side of \eqref{eq:barycentric-total-energy-residual} is therefore sign indefinite and is not a new dissipative mechanism.  The parent kinetic energy differs from the barycentric kinetic energy by $K_d$ according to \eqref{eq:kinetic-energy-decomposition-classI}; on parent Class-II processes satisfying the force identification, equations \eqref{eq:relative-kinetic-flux}--\eqref{eq:relative-kinetic-storage-law} exhibit the missing complementary channel explicitly as relative kinetic storage and transport. Equation~\eqref{eq:relative-kinetic-storage-law} is not asserted as an additional native Class-I evolution law for an arbitrary algebraically closed flux $\bm j$; it records the reversible storage channel that accompanies the exact parent relative inertia.  Restoring that channel requires an enlarged reduced energy density and flux; because $K_d$ depends on $\bm j$, it also introduces time derivatives of the algebraically closed diffusion flux into the total-energy equation.  It is therefore not achieved by merely adding $\mathbf R_d$ to the momentum balance.  After the force identification, the standard Class-I model omits this quadratic storage/transport pair coherently.
\end{proof}

\begin{remark}[Extended single-momentum descendants]
One may keep \eqref{eq:exact-summed-momentum-classI-vars} as an exact single-momentum identity, but then $K_d$ and its transport must also be retained if the total mechanical-energy structure is to remain exact.  Such a model is better regarded as an extended single-momentum descendant with relative-inertia storage, not as the standard Class-I model used here.  Retaining only $\mathbf R_d$ would selectively preserve one consequence of relative inertia while suppressing its conjugate storage channel.
\end{remark}

\subsection{The Class-I diffusion driving force inherited from interaction momentum}

The energy-balance identity \eqref{eq:classI-energy-force-identification} turns the Class-II interaction-force constitutive equation into a Class-I diffusion relation.  Define
\begin{equation}
 \bm X^{\circ}
 \eqdef \bm X_p^{\circ}-\bm X_s^{\circ}
 =\nabla^{\circ}(\mu_p-\mu_s)
 +(\bar s_p-\bar s_s)\nabla T
 \label{eq:classI-X}
\end{equation}
and
\begin{equation}
 \bm\Delta_d
 \eqdef
 \frac{\Div\T_p^d}{\rho_p}
 -\frac{\Div\T_s^d}{\rho_s}.
 \label{eq:classI-Delta-d}
\end{equation}
For constant $\alpha$, the three contributions on the right of \eqref{eq:classI-energy-force-identification} are
\begin{align}
 -\Div\T_{p,0}+\omega_p\Div(\T_{s,0}+\T_{p,0})
 &=-p\nabla\omega_s,
 \label{eq:classI-force-pressure}\\
 -\Div(\alpha\taup)+\omega_p\Div\taup
 &=-M_{sp}\beta\Div\taup,
 \label{eq:classI-force-elastic}\\
 -\Div\T_p^d+\omega_p\Div(\T_p^d+\T_s^d)
 &=-M_{sp}\bm\Delta_d.
 \label{eq:classI-force-viscous}
\end{align}
On the other hand, the fixed Class-II representative gives
\begin{equation}
 \bm m_p
 =-M_{sp}\bm X^{\circ}-p\nabla\omega_s+\bm m_p^d.
 \label{eq:classII-interaction-for-reduction}
\end{equation}
Comparison with \eqref{eq:classI-energy-force-identification} therefore yields the exact algebraic identity
\begin{equation}
 \bm m_p^d=M_{sp}\bm F^{\rm I},
 \label{eq:mpd-classI-force}
\end{equation}
where the inherited Class-I diffusion driving force is
\begin{equation}
 \bm F^{\rm I}
 \eqdef
 \bm X^{\circ}
 -\beta\Div\taup
 -\left(
 \frac{\Div\T_p^d}{\rho_p}
 -\frac{\Div\T_s^d}{\rho_s}
 \right).
 \label{eq:classI-diffusion-force}
\end{equation}
The first term is the ordinary thermo-chemical force.  The second is the configurational-stress contribution inherited from the chosen deformation/stress-decomposition representative.  The third is the partial-viscous-stress contribution familiar from entropy-invariant Class-II to Class-I reduction of ordinary mixtures.

The diagonal Class-II drag closure \eqref{eq:drag-closure} is now transformed, rather than re-postulated.  Since $\bm w=\bm j/M_{sp}$, equations \eqref{eq:drag-closure} and \eqref{eq:mpd-classI-force} give
\begin{equation}
 \bm j
 =-\frac{M_{sp}^2}{\gamma}\,\bm F^{\rm I},
 \qquad \gamma>0.
 \label{eq:classI-diffusion-law}
\end{equation}
The exact parent precursor of this relation is
\begin{equation}
 \bm j=-\frac{M_{sp}^2}{\gamma}
 \left[\bm F^{\rm I}+(\bm a_p-\bm a_s)\right],
 \label{eq:dynamic-diffusion-precursor}
\end{equation}
where $\bm F^{\rm I}$ denotes the same force expression evaluated in the parent fields. Indeed, \eqref{eq:relative-force-residual} and the stress decomposition give $\bm{\mathcal A}=-M_{sp}\bm F^{\rm I}-\gamma\bm w$. Equation~\eqref{eq:classI-diffusion-law} follows after setting the relative-acceleration term to zero. This operation is distinct from dropping the quadratic tensor $\mathbf R_d$ in summed momentum.

Equation~\eqref{eq:classI-diffusion-law} is the binary Class-I diffusion law inherited from the Class-II interaction closure.  It is algebraic in the thermodynamic force but, for general partial viscous closures, is an implicit differential relation for $\bm j$ because $\bm\Delta_d$ contains derivatives of $\T_\alpha^d(\D_\alpha^{\rm I})$.  Thus, for general partial viscous closures, the Class-I diffusion law retains the higher-order stress--diffusion dependence inherited from Class II.

A special case occurs when the dissipative partial stresses are mass-fraction shares of one tensor,
\begin{equation}
 \T_p^d=\omega_p\T^d,
 \qquad
 \T_s^d=\omega_s\T^d,
 \label{eq:classI-mass-fraction-viscous-share}
\end{equation}
then
\begin{equation}
 \bm\Delta_d
 =\left(\frac{\nabla\omega_p}{\rho_p}
 -\frac{\nabla\omega_s}{\rho_s}\right)\!\cdot\T^d,
 \label{eq:classI-Delta-special}
\end{equation}
where $(\nabla\omega)\cdot\T^d$ denotes the vector with components
$\partial_k\omega\,T^d_{ki}$.  In a locally homogeneous-composition state this contribution vanishes.  No such simplification is assumed in the general reduction.

\subsection{Entropy-invariant identification}

The Class-I entropy density is identified with the Class-II entropy density. The entropy flux, however, remains part of the entropy/entropy-flux pair of Principles~I--II and is not fixed solely by the reduced internal-energy flux. Principles~IV--V constrain the production generated by that pair: after reversible contributions have been extracted, it must admit the selected balance-anchored binary-mechanism representation. For the entropy-invariant reduction sought here there is the further requirement that the descendant Class-I production equal the parent production, $\zeta^{\rm I}\equiv\zeta^{\rm II}$. Thus the Class-I entropy flux has to be identified as part of the reduction map rather than assumed a priori. As a first candidate, take the conventional Class-I nonconvective entropy flux associated with \eqref{eq:classI-reduced-heat-def},
\begin{equation}
 \bm\Phi_0^{\rm I}
 =(\bar s_p-\bar s_s)\bm j+\frac{\bm q^{\rm I}}{T},
 \label{eq:classI-entropy-flux-naive}
\end{equation}
Direct exploitation of the Class-I mass balances \eqref{eq:classI-mass} and internal-energy balance \eqref{eq:classI-internal-energy}, using the same free energy and thermodynamic conjugates as in the parent model, then gives the exact preliminary identity
\begin{align}
 T\zeta_0^{\rm I}={}&
 \T^{\rm I}:\D+p\Div\bm v
 -\Zc:\Db\C
 -\bm j\cdot\Big[
 (\nabla\mu_p+\bar s_p\nabla T)
 -(\nabla\mu_s+\bar s_s\nabla T)
 \Big] \notag\\
 &-\frac{1}{T}\bm q^{\rm I}\cdot\nabla T.
 \label{eq:classI-preliminary-entropy}
\end{align}
The corresponding Class-II identity \eqref{eq:exact-entropy-identity}, after only the algebraic substitutions \eqref{eq:classI-velocity-identification}, reads
\begin{align}
 T\zeta^{\rm II}={}&
 \sum_{\alpha=s,p}\T_\alpha:\D_\alpha^{\rm I}
 -\bm m_p\cdot\frac{\bm j}{M_{sp}}
 +p\Div\bm v
 -\Zc:\Db\C \notag\\
 &-\bm j\cdot\Big[
 (\nabla\mu_p+\bar s_p\nabla T)
 -(\nabla\mu_s+\bar s_s\nabla T)
 \Big]
 -\frac{1}{T}(\bm q^{\rm I}+\bm\Psi)\cdot\nabla T.
 \label{eq:classII-entropy-in-classI-vars}
\end{align}
Subtracting \eqref{eq:classI-preliminary-entropy} from \eqref{eq:classII-entropy-in-classI-vars}, and then using \eqref{eq:classI-stress-power-split} together with the energy-force identification \eqref{eq:classI-energy-force-identification}, shows explicitly that the two preliminary productions differ only by the divergence term below.  The required correction follows directly from \eqref{eq:classI-stress-power-split} and \eqref{eq:classI-heat-identification}:
\begin{equation}
 \zeta_0^{\rm I}
 =\zeta^{\rm II}-\Div\!\left(\frac{\bm\Psi}{T}\right).
 \label{eq:classI-zeta-divergence-difference}
\end{equation}
Consequently select the Class-I entropy flux as
\begin{equation}
 \bm\Phi^{\rm I}
 \eqdef\bm\Phi_0^{\rm I}+\frac{\bm\Psi}{T}
 =(\bar s_p-\bar s_s)\bm j
 +\frac{\bm q^{\rm I}+\bm\Psi}{T}.
 \label{eq:classI-entropy-flux-invariant}
\end{equation}
Because $\bm q^{\rm I}+\bm\Psi=\bm q$, this is precisely the Class-II nonconvective entropy flux \eqref{eq:entropy-flux-exact} rewritten in Class-I variables.  Therefore
\begin{equation}
 \zeta^{\rm I}\equiv\zeta^{\rm II}
 \label{eq:classI-zeta-equals-classII}
\end{equation}
identically, not only after constitutive closure and not only to a given asymptotic order.

\begin{remark}[Residual entropy-flux nonuniqueness]
Equation \eqref{eq:classI-zeta-divergence-difference} fixes only the divergence of the entropy-flux correction. The parent-compatible choice \eqref{eq:classI-entropy-flux-invariant} is therefore determined up to an objective admissible vector field $\bm B_\Phi$ with $\Div\bm B_\Phi=0$. Principle~V constrains factors in the production and cannot by itself exclude an addition invisible to that production. Constitutive restrictions, boundary entropy transport and the selected parent flux are needed to fix this residual freedom.
\end{remark}

\begin{theorem}[Entropy-invariant Class-II to Class-I polymer reduction]
\label{thm:entropy-invariant-classI}
Assume the fixed Class-II balance representative used in \Cref{sec:explicit-mixture-gauge}, the conformation transport and reversible-power cancellation of \Cref{thm:cancellation}, and the algebraic identifications \eqref{eq:classI-velocity-identification}, \eqref{eq:classI-total-stress}, \eqref{eq:classI-energy-force-identification}, and \eqref{eq:classI-heat-identification}.  Choose the Class-I entropy flux \eqref{eq:classI-entropy-flux-invariant} selected above.  Then the Class-I entropy production is algebraically identical to the Class-II entropy production for every local thermodynamic process satisfying these identification conditions.  At conformation level it can be written entirely in Class-I variables as
\begin{equation}
 \begin{aligned}
 T\zeta^{\rm I}={}&
 -\frac{\bm q^{\rm I}+\bm\Psi}{T}\cdot\nabla T
 +\T_p^d:\D_p^{\rm I}
 +\T_s^d:\D_s^{\rm I}\\
 &-\bm j\cdot\bm F^{\rm I}
 -\Zc:\mathbf R_C,
 \end{aligned}
 \label{eq:classI-entropy-production}
\end{equation}
with $\bm\Psi$ from \eqref{eq:Psi-nonhydrostatic} and $\bm F^{\rm I}$ from \eqref{eq:classI-diffusion-force}.
\end{theorem}

\begin{proof}
Before reversible extraction, the Class-I entropy identity obtained from \eqref{eq:classI-mass} and \eqref{eq:classI-internal-energy} has the same thermodynamic terms as \eqref{eq:exact-entropy-identity}, but with $\T^{\rm I}:\D$ and $\bm q^{\rm I}$ in place of the constituent stress powers and $\bm q$.  Subtracting it from the Class-II identity, using \eqref{eq:classI-stress-power-split} and \eqref{eq:classI-energy-force-identification}, gives
\[
 T(\zeta^{\rm II}-\zeta_0^{\rm I})
 =\Div\bm\Psi-\frac{\bm\Psi}{T}\cdot\nabla T
 =T\Div\!\left(\frac{\bm\Psi}{T}\right),
\]
which proves \eqref{eq:classI-zeta-divergence-difference}.  The entropy-flux correction in \eqref{eq:classI-entropy-flux-invariant} therefore yields \eqref{eq:classI-zeta-equals-classII}.

It remains to rewrite the already exploited Class-II production.  Equations \eqref{eq:classI-deformation-tensors} give $\D_\alpha=\D_\alpha^{\rm I}$ exactly, while \eqref{eq:mpd-classI-force} and $\bm w=\bm j/M_{sp}$ imply
\[
 -\bm m_p^d\cdot\bm w=-\bm j\cdot\bm F^{\rm I}.
\]
Finally, \eqref{eq:classI-heat-identification} gives $\bm q=\bm q^{\rm I}+\bm\Psi$, and the conformation relaxation term is unchanged.  Substitution into \eqref{eq:classII-entropy} plus \eqref{eq:zetaC-assumption} gives \eqref{eq:classI-entropy-production} term by term.
\end{proof}

\begin{proposition}[Covariance of the entropy-invariant reduction under coordinated representation changes]
\label{prop:classI-gauge-covariance}
Let a Class-II representative be changed according to Proposition~\ref{prop:gauge-transform}, with symmetric tensor field $\mathbf S$, and perform the same algebraic Class-II to Class-I reduction before imposing any new reduced closure.  Then
\begin{align}
 \T^{{\rm I}\prime}&=\T^{\rm I},
 &\bm\Psi'&=\bm\Psi+\mathbf S\bm w,\notag\\
 \bm j_u^{{\rm I}\prime}&=\bm j_u^{\rm I},
 &\bm q^{{\rm I}\prime}&=\bm q^{\rm I}.
 \label{eq:classI-gauge-native-invariants}
\end{align}
The reduced energy--force identification is covariant,
\begin{equation}
 \bm m_p'
 =-\Div\T_p'+\omega_p\Div\T^{{\rm I}\prime},
 \label{eq:classI-gauge-force-covariance}
\end{equation}
and the entropy pair transforms as
\begin{equation}
 \bm\Phi^{{\rm I}\prime}
 =\bm\Phi^{\rm I}+\frac{\mathbf S\bm w}{T},
 \qquad
 \zeta^{{\rm I}\prime}
 =\zeta^{\rm I}+\Div\!\left(\frac{\mathbf S\bm w}{T}\right).
 \label{eq:classI-gauge-entropy-covariance}
\end{equation}
Consequently the entropy-invariant reduction commutes with the coordinated representation change:
\begin{equation}
 \zeta^{\rm I}\equiv\zeta^{\rm II}
 \quad\Longrightarrow\quad
 \zeta^{{\rm I}\prime}\equiv\zeta^{{\rm II}\prime}.
 \label{eq:classI-gauge-commutation}
\end{equation}
If, in addition, $\mathbf S$ is transported entirely within the reversible stress--interaction sector so that $\T_\alpha^d$ and $\bm m_p^d$ are unchanged, then the inherited diffusion force is invariant:
\begin{equation}
 \bm F^{{\rm I}\prime}=\bm F^{\rm I}.
 \label{eq:classI-gauge-force-invariant}
\end{equation}
\end{proposition}

\begin{proof}
The total stress is unchanged because the two stress shifts cancel.  From \eqref{eq:Psi-classI} and $\bm u_p-\bm u_s=\bm w$,
\[
 \bm\Psi'
 =(\T_p+\mathbf S)\bm u_p+(\T_s-\mathbf S)\bm u_s
 =\bm\Psi+\mathbf S\bm w.
\]
Proposition~\ref{prop:gauge-transform} gives $\bm j_u'=\bm j_u+\mathbf S\bm w$ and $\bm q'=\bm q+\mathbf S\bm w$.  Hence
\[
 \bm j_u^{{\rm I}\prime}
 =\bm j_u'-\bm\Psi'
 =\bm j_u-\bm\Psi
 =\bm j_u^{\rm I},
 \qquad
 \bm q^{{\rm I}\prime}
 =\bm q'-\bm\Psi'
 =\bm q^{\rm I}.
\]
Furthermore,
\[
 -\Div\T_p'+\omega_p\Div\T^{{\rm I}\prime}
 =-\Div\T_p-\Div\mathbf S+\omega_p\Div\T^{\rm I}
 =\bm m_p-\Div\mathbf S
 =\bm m_p',
\]
which proves \eqref{eq:classI-gauge-force-covariance}.  Since the preliminary Class-I entropy flux \eqref{eq:classI-entropy-flux-naive} depends on the invariant $\bm q^{\rm I}$, it is unchanged, whereas the entropy-invariance correction $\bm\Psi/T$ acquires $\mathbf S\bm w/T$.  Equation \eqref{eq:classI-gauge-entropy-covariance} therefore follows from the reduced entropy balance and is exactly the image of \eqref{eq:gauge-entropy-pair}.  This proves \eqref{eq:classI-gauge-commutation}.  Finally, if the representation change is confined to the reversible sector, then $\bm m_p^d$ is unchanged.  Because $M_{sp}$ is a state function and \eqref{eq:mpd-classI-force} defines the inherited force invariantly by $\bm m_p^d=M_{sp}\bm F^{\rm I}$, equation \eqref{eq:classI-gauge-force-invariant} follows.
\end{proof}

\begin{remark}[Invariant and representative-dependent reduced quantities]
Under Proposition~\ref{prop:classI-gauge-covariance}, the Class-I total stress, internal-energy flux, heat flux and, for reversible representation changes, the diffusion force are invariant.  The entropy flux and pointwise entropy production inherit the parent divergence shift.  Re-imposing untransformed Fourier, viscous or drag closures after a representation change therefore defines a different constitutive model rather than another representation of the same entropy-invariant reduction.
\end{remark}

For the inherited diagonal closures one obtains explicitly
\begin{equation}
 \begin{aligned}
 \zeta^{\rm I}={}&
 \frac{1}{T^2}(\nabla T)\cdot\mathbf K_T(\nabla T)
 +\sum_{\alpha=s,p}\frac{1}{T}
 \left(2\eta_\alpha|\dev\D_\alpha^{\rm I}|^2
 +\kappa_\alpha^b(\Div\bm v_\alpha^{\rm I})^2\right)\\
 &+\frac{\gamma}{T M_{sp}^2}|\bm j|^2
 +\zeta_C\ge0,
 \end{aligned}
 \label{eq:classI-positive-zeta}
\end{equation}
where
$\bm v_p^{\rm I}=\bm v+\bm j/\rho_p$ and
$\bm v_s^{\rm I}=\bm v-\bm j/\rho_s$.
Since $\bm j=M_{sp}\bm w$, \eqref{eq:classI-positive-zeta} is algebraically identical to \eqref{eq:fully-positive-zeta}.  The inherited local production is unchanged on the stated identification manifold. This statement does not equate entropy histories of solutions to the two different dynamical systems.

The Class-I decomposition in \eqref{eq:classI-entropy-production} is obtained algebraically from the already entropy-exploited Class-II production. No new Principle~V selection is made at this stage; the identity \eqref{eq:classI-zeta-equals-classII} is established before the inherited constitutive relations are rewritten in Class-I variables. In particular, this algebraic reduction does not convert the conformation admissibility condition $-\Zc:\mathbf R_C\ge0$ into a constitutive law for $\mathbf R_C$.

\subsection{Class-I temperature equation: native identity and inherited decomposition}

The Class-I temperature equation admits two distinct forms.  The first follows directly from the reduced mass and internal-energy balances and therefore involves only native Class-I quantities.  The second resolves the same identity into mechanisms inherited from the fixed Class-II representative.

\begin{proposition}[Native Class-I temperature identity]
\label{prop:classI-temperature-native}
Assume the thermodynamic state structure used above and the exact Class-I mass balances \eqref{eq:classI-mass} and internal-energy balance \eqref{eq:classI-internal-energy}.  Then
\begin{equation}
 \begin{aligned}
 \CV\Db T={}&
 -\Div\bm q^{\rm I}
 +\T^{\rm I}:\D
 +(p-Tp_T)\Div\bm v\\
 &-\bm j\cdot\nabla(h_p-h_s)
 -(\Zc-T\Zc_T):\Db\C.
 \end{aligned}
 \label{eq:temperature-classI-native}
\end{equation}
This identity is invariant under the coordinated representation change of Proposition~\ref{prop:classI-gauge-covariance}.
\end{proposition}

\begin{proof}
The thermodynamic differential of the internal energy is
\[
 \dd u
 =\CV\,\dd T
 +\sum_{\alpha=s,p}h_\alpha\,\dd\rho_\alpha
 +(\Zc-T\Zc_T):\dd\C.
\]
Using the Class-I mass balances in barycentric form,
$\Db\rho_\alpha+\rho_\alpha\Div\bm v+\Div\bm J_\alpha=0$ with $\bm J_p=\bm j=-\bm J_s$, together with
$\sum_\alpha\rho_\alpha h_\alpha=u+p-Tp_T$, and inserting the result into \eqref{eq:classI-internal-energy}, all terms containing $\Div\bm J_\alpha$ combine with the relative-energy part of $\bm j_u^{\rm I}$.  The remaining diffusion contribution is $-\bm j\cdot\nabla(h_p-h_s)$, which gives \eqref{eq:temperature-classI-native}.  The invariance statement follows because $\T^{\rm I}$ and $\bm q^{\rm I}$ are invariant by \eqref{eq:classI-gauge-native-invariants}, while the state functions and reduced kinematic variables are unchanged.
\end{proof}

To expose the inherited dissipative and thermoelastic mechanisms, return now to the fixed Class-II representative used in the parent entropy exploitation.  Define
\begin{equation}
 \Phi_d^{\rm I}
 \eqdef
 \T_p^d:\D_p^{\rm I}
 +\T_s^d:\D_s^{\rm I}
 -\bm j\cdot\bm F^{\rm I}.
 \label{eq:classI-Phi-d}
\end{equation}
By \eqref{eq:mpd-classI-force} and \eqref{eq:classI-velocity-identification},
$\Phi_d^{\rm I}\equiv\Phi_d$ algebraically.  Using
$\bm q=\bm q^{\rm I}+\bm\Psi$,
$\bm J_p=\bm j=-\bm J_s$ and
$\D_d=\D_d^{\rm I}$ in the already derived Class-II temperature equation yields the equivalent inherited decomposition
\begin{equation}
 \begin{aligned}
 \CV\Db T={}&
 -\Div(\bm q^{\rm I}+\bm\Psi)
 +\Phi_d^{\rm I}
 -T p_T\Div\bm v\\
 &-T\bm j\cdot\nabla^{\circ}(\bar s_p-\bar s_s)
 +T\left(\frac{\partial\taup}{\partial T}\right)_{\rho_s,\rho_p,\C}:\D_d^{\rm I}\\
 &-(\Zc-T\Zc_T):\mathbf R_C.
 \end{aligned}
 \label{eq:temperature-classI}
\end{equation}
Equation \eqref{eq:temperature-classI} is algebraically equivalent to the native identity \eqref{eq:temperature-classI-native} in the fixed representative and displays diffusion heating, inherited stress--diffusion transport, relative constituent entropy transport, reversible thermoelastic heating/cooling and configurational relaxation heat separately.

\begin{remark}[Representation covariance of the thermal decomposition]
Equation \eqref{eq:temperature-classI-native} is representation invariant.  In \eqref{eq:temperature-classI}, $\bm\Psi$ and the corresponding reversible-power terms transform together, leaving the complete temperature identity unchanged.  Thus $-\Div(\bm q^{\rm I}+\bm\Psi)$ is a representative-dependent inherited decomposition; the native Class-I heat transport is $-\Div\bm q^{\rm I}$.
\end{remark}

\subsection{Closure status of the Class-I descendant}

The entropy-invariant reduction transports the parent constitutive information but does not create missing constitutive laws. With the parent free energy, heat conductivity, viscous responses, drag coefficient, and an explicit objective conformation-source law $\mathbf R_C=\mathcal R_C$ all specified, with $-\Zc:\mathcal R_C\ge0$, the descendant is constitutively closed at the level of equation counting, provided the chosen source also preserves the admissible conformation state space. The Hookean law has this property by Proposition~\ref{prop:SPD-preservation}. If only the inequality $-\Zc:\mathbf R_C\ge0$ is imposed, the descendant is thermodynamically admissible but not constitutively closed because $\mathbf R_C$ remains undetermined. No claim of PDE well-posedness or sufficiency of boundary conditions is made here. The unknown fields are
\[
 \rho_s,\qquad \rho_p,\qquad \bm v,\qquad \bm j,\qquad T,\qquad \C,
\]
with the algebraic constraint $\rho=\rho_s+\rho_p$.  The governing equations are the two mass balances \eqref{eq:classI-mass}, the single momentum balance \eqref{eq:classI-momentum}, the conformation equation \eqref{eq:classI-conformation}, and the native temperature identity \eqref{eq:temperature-classI-native}; in the fixed inherited representative the latter is algebraically equivalent to the mechanism-resolved form \eqref{eq:temperature-classI}.  Their inherited constitutive relations are
\begin{align}
 \T^{\rm I}&=-p\I+\taup
 +\T_p^d(\D_p^{\rm I})+\T_s^d(\D_s^{\rm I}),
 \label{eq:classI-closed-stress}\\
 \bm j&=-\frac{M_{sp}^2}{\gamma}\bm F^{\rm I},
 \label{eq:classI-closed-diffusion}\\
 \bm q^{\rm I}&=-\mathbf K_T\nabla T-\bm\Psi,
 \label{eq:classI-closed-heat}\\
 \mathbf R_C&=\mathcal R_C,
 &-\Zc:\mathcal R_C&=T\zeta_C\ge0.
 \label{eq:classI-closed-conformation-dissipation}
\end{align}
with $\bm F^{\rm I}$, $\bm\Psi$, $\D_p^{\rm I}$, $\D_s^{\rm I}$ and $\D_d^{\rm I}$ defined in \eqref{eq:classI-diffusion-force}, \eqref{eq:Psi-nonhydrostatic} and \eqref{eq:classI-deformation-tensors}. Here $\mathcal R_C$ denotes the explicitly chosen objective conformation-source constitutive map, not an additional unknown. For the Hookean model, $\mathcal R_C=-\lambda^{-1}(\C-\C_{\rm eq})$ by \eqref{eq:C-hookean}; the Gaussian restriction is required for the dissipation match of Proposition~\ref{prop:gaussian-dissipation-match}, not for closure of the Hookean second moment itself.

Equation \eqref{eq:classI-closed-heat} is the transported heat-flux closure in the fixed parent representative.  Under a coordinated reversible rewrite, the parent heat flux and $\bm\Psi$ acquire the same $\mathbf S\bm w$ increment, while $\bm q^{\rm I}$ remains invariant by Proposition~\ref{prop:classI-gauge-covariance}.

The resulting PDE system is not necessarily second order in space.  In the generic case, $\bm F^{\rm I}$ contains $\Div\T_\alpha^d$, whereas $\T_\alpha^d$ itself depends on gradients of $\bm j/\rho_\alpha$ through \eqref{eq:classI-deformation-tensors}.  Hence the reduced diffusion relation retains the higher-order stress--diffusion terms inherited from the Class-II constitutive laws.

\subsection{Exact identities and the relative-inertia truncation}

The reduction contains five logically distinct statements.
\begin{enumerate}[label=(\roman*)]
\item The changes of variables \eqref{eq:classI-velocity-identification}, the partial mass balances \eqref{eq:classI-mass}, the conformation equation \eqref{eq:classI-conformation}, and the stress/flux rearrangements are exact changes of variables.
\item The vectorial condition \eqref{eq:classI-energy-force-identification} is an additional constitutive identification, equivalent to elimination of relative acceleration on parent processes. The internal-energy and entropy-production identities are algebraic consequences conditional on that choice, not unconditional identities between solution spaces.
\item The standard Class-I momentum equation \eqref{eq:classI-momentum} neglects the relative-inertia flux $\mathbf R_d=\bm j\otimes\bm j/M_{sp}$ relative to the exact summed Class-II momentum balance \eqref{eq:exact-summed-momentum-classI-vars}.  At the same quadratic order the standard Class-I mechanical-energy architecture omits the relative kinetic-energy storage $K_d=|\bm j|^2/(2M_{sp})$.  Proposition~\ref{prop:relative-inertia-consistency} shows that these two omissions must be treated together: keeping the momentum flux alone would leave the uncompensated reversible power $\mathbf R_d:\D$.
\item The relative-inertia truncation does not alter the exact algebraic identity $\zeta^{\rm I}\equiv\zeta^{\rm II}$.  On the identification manifold, the quadratic storage/transport channel is reversible and of order $O(|\bm w|^2)$; the inherited drag production is still retained. This must not be read as declaring the eliminated relative-momentum dynamics nondissipative.  This is the precise sense in which entropy invariance and an approximate momentum balance are compatible.
\item The subsequent specialization $\bm j\to\bm0$ is a further physical reduction.  It eliminates diffusion itself and its entropy production; it is not part of the entropy-invariant Class-II to Class-I identification.
\end{enumerate}

\section{From Class I to the one-velocity model and classical upper-convected rheology}
\label{sec:reductions}

The entropy-invariant Class-I model retains a nonzero diffusion flux. The one-velocity model is obtained by a further reduction.  Consider a nondimensional fast-drag scaling in which the reciprocal drag strength is a small parameter, the constituent densities remain bounded away from zero, and the remaining Class-I driving forces and temperature stay uniformly bounded on the scale of interest. Assume also the spatial derivative bounds needed to make the velocity-gradient corrections small and, when comparing with Class II, exclude unresolved fast initial layers in relative acceleration. These are hypotheses of a formal limit, not estimates proved here.  From the reduced relation \eqref{eq:classI-diffusion-law},
\begin{equation}
 \bm w=\bm v_p-\bm v_s=O(\gamma^{-1}),
 \qquad
 \bm j=M_{sp}\bm w=O(\gamma^{-1}),
 \label{eq:fast-drag}
\end{equation}
so that $\mathbf R_d=O(\gamma^{-2})$ and $K_d=O(\gamma^{-2})$, whereas the drag entropy production $\gamma|\bm w|^2/T$ is $O(\gamma^{-1})$.  Thus, in this nondimensional scaling, the coherent relative-inertia pair omitted from the standard Class-I momentum/energy architecture is $O(\gamma^{-2})$, whereas the retained drag production is $O(\gamma^{-1})$; the former is asymptotically smaller in the fast-drag limit.  To leading order,
\begin{equation}
 \bm v_p=\bm v_s=\bm v,
 \qquad
 \Lgrad_d=\nabla\bm v
 \label{eq:one-velocity}
\end{equation}
for every constant $\alpha$. Equation \eqref{eq:two-velocity-upper} reduces to the standard upper-convected derivative
\begin{equation}
 \overset{\nabla}{\C}
 =\partial_t\C+\bm v\cdot\nabla\C
 -(\nabla\bm v)\C-\C(\nabla\bm v)^{\sf T}.
 \label{eq:standard-upper}
\end{equation}
The two-fluid ambiguity is therefore invisible in a homogeneous one-velocity model.

For Hookean dumbbells, \eqref{eq:C-hookean} becomes
\begin{equation}
 \lambda\overset{\nabla}{\C}+\C=\frac{k_BT}{H}\I.
 \label{eq:onefluid-C}
\end{equation}
The stress equation \eqref{eq:compressible-nonisothermal-stress} becomes
\begin{equation}
 \overset{\nabla}{\taup}
 +\left(\frac1\lambda+\Div\bm v-\dot{\ln H}\right)\taup
 =2nk_BT\D
 +nk_BT\left(\dot{\ln H}-\dot{\ln T}\right)\I.
 \label{eq:onefluid-compressible-stress}
\end{equation}
Equation \eqref{eq:onefluid-compressible-stress} is the non-isothermal one-fluid reduction under the stated assumptions. The shorter form
\[
 \overset{\nabla}{\taup}+\left(\lambda^{-1}+\Div\bm v\right)\taup
 =2nk_BT\D-nk_B\dot T\,\I
\]
requires the additional assumption $\dot H=0$. It must not be used when $H(T)$ is retained as a thermoelastic material function.
Only after the additional reductions
\begin{equation}
 \Div\bm v=0,
 \qquad
 T=T_0,
 \qquad
 n=n_0>0,
 \label{eq:isoincomp}
\end{equation}
can one write
\begin{equation}
 \lambda\overset{\nabla}{\taup}+\taup=2\eta_p\D,
 \qquad
 \eta_p=nk_BT\lambda.
 \label{eq:UCM}
\end{equation}
which is the upper-convected Maxwell polymer stress. Adding a Newtonian solvent stress gives the Oldroyd-B structure. Uniform $n_0$ is an additional preparation assumption: incompressibility only implies material conservation of $n$, not spatial uniformity.

An imposed isothermal rheological model is not obtained by retaining an arbitrary dissipative flow together with the unchanged adiabatic energy balance and setting $T=T_0$. Its temperature equation is replaced by a prescribed-temperature constraint and a compatible thermal reservoir or heat supply. In particular, maintaining globally uniform temperature during general dissipative motion requires a compensating thermal exchange; it is not the generic evolution of the isolated parent model.

\subsection{Information lost under early incompressible/isothermal reduction}

Three kinds of information are lost if incompressibility/isothermality are imposed before entropy exploitation.
\begin{enumerate}[label=(\roman*)]
\item Volumetric terms distinguish specific and density-weighted conformation rates, cf. \eqref{eq:truesdell-M}.
\item Isotropic conformation stresses can be hidden inside the incompressibility pressure although they are thermodynamically determined in the compressible parent model.
\item Temperature enters both the equilibrium conformation and the polymer stress scale; normalization by $T$ creates the explicit thermal rate term \eqref{eq:A-thermal}.
\end{enumerate}
Thus imposing incompressibility and isothermality before entropy exploitation can hide thermodynamically determined isotropic stresses and removes the volumetric and thermal information used in the full parent derivation. The objective-rate compatibility test itself does not require those reductions to be absent; the point is that carrying them out only afterwards preserves the larger set of constraints supplied by the compressible, non-isothermal Class-II parent.

\section{Relation to existing polymer models}
\label{sec:literature-comparison}

Two-fluid polymer theories contain several structures relevant to the reductions above. Stress--composition coupling and stress-driven migration are treated in Doi--Onuki and Apostolakis--Mavrantzas--Beris~\cite{DoiOnuki1992,Apostolakis2002}. Thermodynamic two-fluid formulations based on generalized brackets, free-energy dissipation or Onsager structure include Mavrantzas--Beris, Beris--Mavrantzas, Beris--Edwards, Zhou--Zhang--E, Hooshyar--Germann and Spiller et al.~\cite{MavrantzasBeris1992,BerisMavrantzas1994,BerisEdwards1994,ZhouZhangE2006,HooshyarGermann2016,Spiller2021}. In particular, Spiller et al.~\cite{Spiller2021} first retain inertial dynamics of macromolecular and relative two-fluid motion and subsequently eliminate those inertial variables; this is a close precedent for treating relative inertia as a separate reduction issue, although their coarse-grained GENERIC construction is not the entropy-invariant Class-II to Class-I identification used here. Relations between convective/deformation velocities and the sharing of polymer stress are discussed by Tanaka, Pleiner--Harden and Zhou--Doi~\cite{Tanaka1997,PleinerHarden2004,ZhouDoi2022}. Kinetic two-fluid descriptions with distinct solvent and polymer velocities are given by Degond--Lozinski--Owens and Singh--Subramanian--Ansumali~\cite{Degond2010,Singh2022}.

Thermodynamically consistent one-fluid conformation models have been derived by generalized-bracket/GENERIC, continuum and microstructural methods~\cite{GrmelaOttinger1997}. A recent synthesis by Theodorou--Mavrantzas reviews the generalized-bracket and GENERIC routes and, for isothermal polymer solutions, develops a two-fluid hydrodynamic model with stress-induced migration~\cite{TheodorouMavrantzas2025}. Non-isothermal formulations distinguish stored configurational energy from entropy production and, in compressible settings, distinguish different conformation-tensor normalizations~\cite{Guaily2015,Guaily2020,Dostalik2020}. Equations \eqref{eq:compressible-nonisothermal-stress} and \eqref{eq:truesdell-M} exhibit the corresponding dependence of tensor transport on compressibility and state-variable choice.

The closest structural overlaps should therefore be separated carefully. Bothe--Dreyer provide the parity-based entropy principle and entropy-invariant Class-II to Class-I reduction for molecular mixtures, but without polymer conformation~\cite{BotheDreyer2015}. Zhou--Doi derive isothermal two-fluid polymer equations from Onsager's principle and show that stress-driven diffusion depends on the chosen sharing of deformation kinematics~\cite{ZhouDoi2022}; their dilute free energy is likewise linear in polymer concentration, so the elastic term does not contribute to osmotic pressure. Dostal\'ik--M\'alek--Pr\r{u}\v{s}a--S\"uli derive a compressible non-isothermal dilute-polymer kinetic model and temperature equation in a one-fluid/solvent-dominated setting~\cite{Dostalik2020}. At mixture level, Prahs recently emphasized that apparent Onsager thermo-diffusive closure structure depends on the representation of diffusion-related transport and on the chosen thermodynamic-force basis~\cite{Prahs2026}. Both results therefore expose representation dependence, although the result of Prahs concerns a different mixture-level energy/entropy-flux repartition rather than the coordinated constituent stress--interaction changes or polymer conformation considered here. The present claims concern the intersection of these structures: a common-temperature Class-II polymer mixture with two constituent momenta, explicit reversible configurational-power cancellation, representation covariance, and entropy-invariant reduction to a diffusive Class-I descendant.

\section{Conclusions}

For the Class-II solvent--polymer model, the total entropy identity permits configurational and mechanical exchange terms to be separated before constitutive sign conditions are imposed. For the dilute conformation free energy considered here, conformation-dependent chemical-potential transport cancels the barycentric-to-polymer convection correction, and the remaining configurational deformation power cancels the elastic partial-stress power when the weights defining the deformation velocity and elastic-stress decomposition match. The residual entropy production contains heat conduction, partial viscous dissipation, interconstituent drag and configurational relaxation. The corresponding signed term in the reduced configurational relative-free-energy balance is therefore an internal mechanical exchange rather than a separate entropy production.

The kinetic distribution description and the conformation-tensor description remain distinct levels of description. For Hookean springs the second-moment evolution closes exactly without a Gaussian hypothesis, but their relaxation productions agree through the finite-dimensional conformation free energy only for the centered Gaussian Hookean invariant class proved in Proposition~\ref{prop:gaussian-dissipation-match}; the non-Gaussian free-energy and dissipation remainders in Proposition~\ref{prop:nongaussian-remainder} are nonnegative and quantify the information omitted by the moment state. The connector population balance fixes the polymer material derivative and yields the two-velocity upper-convected conformation transport by moments; objectivity is a covariance requirement, not the selector of that transport law. The matching relation between deformation velocity and elastic-stress decomposition follows from cancellation of reversible configurational and mechanical powers in the fixed Class-II representative. For the same dumbbell free energy, Kramers stress and no additional reversible channel, the stated mechanism-wise Gordon--Schowalter compatibility test requires the affine upper-convected choice $a=1$; alternative objective rates require a changed state variable or compensating reversible physics.

The Class-II to Class-I constitutive identification preserves the entropy production algebraically on its identification manifold. It eliminates relative acceleration and is not a dynamical equivalence theorem. The Class-I diffusion force inherits thermo-chemical, configurational-stress and partial-viscous-stress terms, and the selected Class-I entropy flux gives $\zeta^{\rm I}\equiv\zeta^{\rm II}$. Coordinated reversible stress--interaction changes leave the native Class-I stress and heat flux invariant and transform the local Class-I entropy pair by the same divergence as the parent pair. The exact summed momentum equation contains the relative-inertia flux $\mathbf R_d=\bm j\otimes\bm j/M_{sp}$. Its associated storage $K_d=|\bm j|^2/(2M_{sp})$ and flux $\bm Q_d$ satisfy the purely reversible specialization in Proposition~\ref{prop:relative-inertia-consistency} only on parent processes that also obey the relative-acceleration identification; the standard Class-I model therefore omits the relative-inertia momentum flux together with its reversible kinetic-energy channel. This truncation does not alter the algebraic entropy-production identity.

The Class-II and Class-I temperature equations separate conductive transport, viscous conversion, polymer--solvent drag, thermo-compressive work, relative constituent transport, thermoelastic heating or cooling and configurational relaxation. The relaxation contribution to entropy production is not in general identical to the local thermal source because the configurational free energy depends on temperature. The non-isothermal Hookean stress equation likewise retains the temperature dependence of the spring stiffness and equilibrium conformation. Temperature-dependent normalization of that conformation requires both kinematic and thermodynamic chain-rule corrections; the stress scale $nk_BT$ alone does not imply purely entropic thermal response.

The zero-diffusion limit converts the Class-I model to a one-velocity conformation theory. Incompressibility, isothermality and constant polymer number density then yield the upper-convected Maxwell/Oldroyd-B equations. These reductions remove constitutive information present in the Class-II parent and are therefore applied only after the entropy and energy identities have been established.

\section*{Acknowledgment}
The author gratefully acknowledges discussions with the late Wolfgang Dreyer on entropy principles and non-Newtonian constitutive modeling.

\appendix
\section{Detailed derivation of the reduced relative-free-energy identity}
\label{app:relative-entropy}

Assume at the selected local test process that $\Dp T=0$. From \eqref{eq:phi-eq},
\begin{equation}
 \nabla_{\bm r}\ln\frac{\varphi}{\varphi_{\rm eq}}
 =\nabla_{\bm r}\ln\varphi+\frac1{k_BT}\nabla_{\bm r}U.
 \label{eq:relative-grad}
\end{equation}
Using \eqref{eq:phi-balance}, normalization and integration by parts,
\begin{align}
 \partial_t\eta_c+\Div(\eta_c\bm v_p)
 &=-nk_B\int \Dp\varphi\,
 \ln\frac{\varphi}{\varphi_{\rm eq}}\,\dd\bm r\notag\\
 &=-nk_B\int\left(\varphi\Lgrad_d\bm r+\Jr\right)\cdot
 \nabla_{\bm r}\ln\frac{\varphi}{\varphi_{\rm eq}}\,\dd\bm r.
 \label{eq:relative-step1}
\end{align}
The affine term equals
\begin{align}
 &-nk_B\int\varphi(\Lgrad_d\bm r)\cdot\nabla_{\bm r}\ln\varphi\,\dd\bm r
 -\frac nT\int\varphi(\Lgrad_d\bm r)\cdot\nabla_{\bm r}U\,\dd\bm r\notag\\
 &\qquad=nk_B\tr\Lgrad_d
 -\frac1T\left(\taup:\Lgrad_d+nk_BT\tr\Lgrad_d\right)
 =-\frac1T\taup:\Lgrad_d.
 \label{eq:relative-affine}
\end{align}
For \eqref{eq:Jr-closure}, the connector-flux term is
\begin{equation}
 -nk_B\int\Jr\cdot\nabla_{\bm r}\ln\frac{\varphi}{\varphi_{\rm eq}}\,\dd\bm r
 =-\frac nT\int\Jr\cdot\nabla_{\bm r}\mu_c\,\dd\bm r
 =\zeta_r.
 \label{eq:relative-diss}
\end{equation}
This proves \eqref{eq:partial-config-entropy}.

\section{Derivation of the Hookean moment equation}
\label{app:hookean}

Multiply \eqref{eq:phi-balance} by $\bm r\otimes\bm r$ and integrate. The affine term gives
\begin{equation}
 -\int \bm r\otimes\bm r\,\nabla_{\bm r}\cdot(\varphi\Lgrad_d\bm r)\,\dd\bm r
 =\Lgrad_d\C+\C\Lgrad_d^{\sf T}.
 \label{eq:moment-affine}
\end{equation}
For \eqref{eq:hookean-Jr},
\begin{align}
 \int(\Jr\otimes\bm r+\bm r\otimes\Jr)\,\dd\bm r
 &=-\frac{4H}{\zeta_b}\C
 -\frac{2k_BT}{\zeta_b}
 \int\left(\nabla_{\bm r}\varphi\otimes\bm r+\bm r\otimes\nabla_{\bm r}\varphi\right)\dd\bm r\notag\\
 &=-\frac{4H}{\zeta_b}\C+\frac{4k_BT}{\zeta_b}\I,
 \label{eq:moment-relax}
\end{align}
which yields \eqref{eq:C-hookean-raw}.

\section{Derivation of the compressible non-isothermal stress equation}

Starting from
\[
 \taup=nH(T)\C-nk_BT\I,
\]
apply the operator $\overset{\nabla(p|d)}{(\cdot)}=\Dp(\cdot)-\Lgrad_d(\cdot)-(\cdot)\Lgrad_d^{\sf T}$. Then
\begin{align}
 \overset{\nabla(p|d)}{\taup}
 ={}&nH\overset{\nabla(p|d)}{\C}
 +H(\Dp n)\C
 +n(\Dp H)\C \notag\\
 &-k_B\Dp(nT)\I
 +2nk_BT\D_d.
 \label{eq:stress-derivation1}
\end{align}
Use \eqref{eq:C-hookean}, $\Dp n=-n\Div\bm v_p$ and
\[
 H\C=\frac{\taup}{n}+k_BT\I.
\]
The isotropic terms proportional to $T\Div\bm v_p$ cancel. The terms proportional to $\Dp H$ are separated into a multiple of $\taup$ and an isotropic contribution. This gives
\[
 \overset{\nabla(p|d)}{\taup}
 +\left(\frac1\lambda+\Div\bm v_p-\Dp\ln H\right)\taup
 =2nk_BT\D_d
 +nk_BT\left(\Dp\ln H-\Dp\ln T\right)\I,
\]
which is \eqref{eq:compressible-nonisothermal-stress}.

\section{Detailed derivation of the temperature equation}
\label{app:temperature}

The derivation of \Cref{thm:temperature} follows directly from the energy equation. From $u=f+Ts$, $s=-f_T$, $\mu_\alpha=f_{\rho_\alpha}$ and $\Zc=f_\C$,
\begin{equation}
 u_T=-Tf_{TT}=\CV,
 \qquad
 u_{\rho_\alpha}=\mu_\alpha-Tf_{T\rho_\alpha}
 =\mu_\alpha+T\bar s_\alpha=h_\alpha,
 \qquad
 u_\C=\Zc-T\Zc_T.
 \label{eq:u-derivatives-app}
\end{equation}
Thus
\begin{equation}
 \Db u
 =\CV\Db T
 +\sum_\alpha h_\alpha\Db\rho_\alpha
 +(\Zc-T\Zc_T):\Db\C.
 \label{eq:Du-app}
\end{equation}
Use
\[
 \Db\rho_\alpha=-\rho_\alpha\Div\bm v-\Div\bm J_\alpha
\]
and $\bm J_u=u\bm v+\bm q+\sum_\alpha h_\alpha\bm J_\alpha$ in \eqref{eq:internal-energy-exact}. The coefficient of $\Div\bm v$ becomes
\begin{align}
 u-\sum_\alpha\rho_\alpha h_\alpha
 &=f+Ts-\sum_\alpha\rho_\alpha(\mu_\alpha+T\bar s_\alpha)\notag\\
 &=-p+T\left(s-\sum_\alpha\rho_\alpha\bar s_\alpha\right)\notag\\
 &=-p+Tp_T,
 \label{eq:pressure-maxwell-app}
\end{align}
where $p_T=s-\sum_\alpha\rho_\alpha\bar s_\alpha$ follows from differentiating \eqref{eq:u-p-def}. This yields \eqref{eq:temperature-raw}.

Next split
\[
 \nabla h_\alpha
 =\nabla\mu_\alpha+\bar s_\alpha\nabla T+T\nabla\bar s_\alpha.
\]
The first two terms are exactly those appearing in the entropy identity \eqref{eq:exact-entropy-identity}. After the reversible Class-II extraction and the elastic cancellation theorem their contribution reduces to $\Phi_d-\Zc:\mathbf R_C$. The remaining term is $-T\sum\bm J_\alpha\cdot\nabla\bar s_\alpha+T\Zc_T:\Db\C$.

For \eqref{eq:dilute-f-split},
\[
 (\bar s_p)_\C=-\frac1{\rho_p}\Zc_T,
 \qquad
 (\bar s_s)_\C=0.
\]
Hence the $\C$-gradient part of $-T\bm J_p\cdot\nabla\bar s_p$ is
\[
 +T\Zc_T:\left[(\bm J_p/\rho_p)\cdot\nabla\C\right],
\]
which cancels the last term in \eqref{eq:C-bary-kinematics} when $T\Zc_T:\Db\C$ is inserted. By isotropy, $\C\Zc_T=\Zc_T\C$, so the deformation part is
\[
 T\Zc_T:(\Lgrad_d\C+\C\Lgrad_d^{\sf T})
 =T(\C\Zc_T+\Zc_T\C):\D_d
 =T(\taup)_T:\D_d,
\]
and the relaxation part combines as
\[
 -\Zc:\mathbf R_C+T\Zc_T:\mathbf R_C
 =-(\Zc-T\Zc_T):\mathbf R_C.
\]
This proves \eqref{eq:temperature-twofluid}.

\end{document}